\documentclass[11pt,a4paper]{article}

\usepackage[cp1250]{inputenc}   
\usepackage[T1]{fontenc}

\usepackage{url}      

\usepackage{amsmath}
\usepackage{mathrsfs} 
\usepackage{amssymb}

\usepackage{manfnt}

\usepackage{bbm}

\usepackage{theorem}  

\usepackage[hidelinks]{hyperref} 

\def\R{\mathbb{R}}
\def\F{\mathbb{F}}
\def\E{\mathbbm{E}}
\def\N{\mathbb{N}}
\def\P{\mathbb{P}}

\def\P{\mathbb{P}}
\def\Q{\mathbb{Q}}
\def\Z{\mathbb{Z}}

\def\AA{\mathcal{A}}
\def\BB{\mathcal{B}}
\def\CC{\mathcal{C}}
\def\FF{\mathcal{F}}
\def\XX{\mathcal{X}}

\def\NN{\mathcal{N}}

\def\GG{\mathcal{G}}

\def\c{\mathscr{C}}

\def\NA{\mathrm{NA}}

\def\hypo{\mathrm{hypo}\,}
\def\ri{\mathrm{ri}}

\def\indic{\mathbf{1}}

\DeclareMathOperator*{\esssup}{ess\,sup}

{
  \theorembodyfont{\itshape}

  \newtheorem{theorem}{Theorem}[section]

  \newtheorem{lemma}      [theorem]{Lemma}
  
  \newtheorem{corollary}  [theorem]{Corollary}
  
}

{
  \theorembodyfont{\rmfamily}

  \newtheorem{definition}[theorem]{Definition}
  \newtheorem{example}   [theorem]{Example}

  \newtheorem{remark}    [theorem]{Remark}
}

\newenvironment{proof}
                {\noindent\textit{Proof.\;}}
                {%
                    \nopagebreak
                    \hfill{\vrule width 1ex height 1ex depth 0ex}
                    \medskip
                    \goodbreak
                }

\usepackage[mathscr]{euscript}

\usepackage{color}
\title{Financial market models in discrete time beyond the concave case
\author{
  Mario \v{S}iki\'c%
  \thanks{
    Department of Mathematics, ETH Zurich, \texttt{mario.sikic@math.ethz.ch}.
    }
 }}

\begin{document}

\maketitle

\begin{abstract}
  In this article we propose a study of market models starting from a set of axioms, as one does in the case of risk measures. We define a market model simply as a mapping from the set of adapted strategies to the set of random variables describing the outcome of trading. We do not make any concavity assumptions. The first result is that under sequential upper-semicontinuity the market model can be represented as a normal integrand. We then extend the concept of no-arbitrage to this setup and study its consequences as the super-hedging theorem and utility maximization. Finally, we show how to extend the concepts and results to the case of vector-valued market models, an example of which is the Kabanov model of currency markets.
\end{abstract}

\section{Introduction}
Thinking in the field of mathematical finance and, more broadly,  economics consists usually of two steps. The first step is postulating a model for the effect that is under study and justification of the model. This consists of connecting the action of the trader to his outcome of trading. From that point on the study focuses on the particular model at hand and analyzes the consequences. 

In mathematical finance a model is a description of the behavior of some indicator of performance in the market as it depends on the action of the agent or trader. The three main problems in mathematical finance are study of whether the model makes sense, i.e. arbitrage theory, hedging, and the study of how to achieve optimal performance. In general the methods themselves depend on the particular model under study.

The ubiquitous assumption of concavity of financial market models comes into the argument due to two things: first, one usually aims at obtaining the dual representation of the set of superhedgeable claims; and second, the space of admissible strategies $\AA$ in the market is not included in a locally convex space, but due to convexity, one can use the trick of passing to 'subsequences of convex combinations'. This line of argument was introduced by Schachermayer~\cite{schachermayer1992hilbert}.

Staying with the concavity assumption, it was in  Pennanen~\cite{pennanen2011convex} that was shown that the questions of no-arbitrage and closure of the set of superhedgeable claims can be embedded into the convex duality framework, thus yielding those results for a much larger set of market models. However, the problem with general models is that with more generality one loses the interpretation. For instance, in the fricionless market model one can state the fundamental theorem of asset pricing as being equivalent to the existence of an equivalent martingale measure for the stock price process. Only in the case of market models with proportional transaction costs is such a clean interpretation still possible (see Kabanov et al.~\cite{kabanov2002no}).
  
Market models with transaction costs that are not convex were considered in Lobo et al.~\cite{lobo2007portfolio}, where proportional and fixed costs are considered in a computation framework. They propose a heuristic, i.e. an iterative scheme, for calculating the optimal strategy and present numerical experiments of those. See also the recent publication by Dolinsky and Kifer~\cite{dolinsky2014risk}.

In this paper we will study models of financial markets that do not satisfy the concavity property. There are two levels of non-concavity in the market model. On one side, it may happen that the set of admissible strategies is not convex. A simple example of this situation is a probably realistic case where one can hold in the portfolio only an integer number of any one asset; this one can model with strategies $\vartheta$ taking values in the set of integers $\Z^d$. On the other side, as noted in F\"ollmer and Schied~\cite{follmer2011stochastic}, transaction costs might not be convex in the sense that a bigger trader could be able to procure a better price than the smaller one per unit volume. Fixed transaction costs is a notable example of this situation. Also, one may notice that some models of market illiquidity are not convex, see e.g. Roch and Soner~\cite{roch2011resilient}.

In this paper we consider the set of strategies available to the trader $\AA$ as a basic object. On a filtered probability space $(\Omega,\FF,\F,P)$ the financial market model will be defined as a mapping from the set of adapted sequences
  $\AA 
    = \{(\vartheta_t)_{t=0,\ldots,T-1}
      \,|\,
      \vartheta_t\in L^0(\FF_t;\R^d)\}$ 
that represents the strategies of the agent in the market to the set $L^0(\FF)$ of measurable random variables representing the outcomes of trading. The values of the market model can, in general, be in any set. For simplicity we consider only two cases: when the market model takes values in the set of random variables $L^0(\FF;\R)$ and when it takes values in the set of random vectors $L^0(\FF;\R^n)$. The canonical variable one considers as a market model is mark-to-market value of the final position, liquidated final wealth, or utility from consumption.
  
  Comparing the research of financial market models with that of risk measures, the difference in approaches is immediately apparent. In the latter one starts with axiomatics: a (conditional) risk measure is a mapping from the space of financial positions in the space $L^\infty(\FF;\R)$ into the space of $L^\infty(\GG;\R)$, which describes the required wealth to make the position acceptable. Here, of course, one has $\GG\subset\FF$. One approaches it this way because it is a general. Interestingly enough, the market models in the financial literature are analyzed on the case by case basis. The main result in the risk measures stream of literature is that under lower semicontinuity property, i.e. the Fatou property, one also has a Fenchel--Moreau dual representation (See F\"ollmer and Schied~\cite{follmer2011stochastic}). 

  One additional comparison of our work is with the recently developed field of $L^0$ modules (see~\cite{cheridito2012conditional}). The $L^0$ module theory developed in the mentioned paper can be equaled directly to our one step market model. We postpone discussion about this connection for later research.
  
  The paper is structured as follows: in Section~2 we define the market model, i.e. lay out the axioms, and explain the reasons for studying those through examples; in Section~3 we state a representation result for the market model; in Section~4 we define the no-arbitrage condition and state its consequences like the closure of the set of superhedgeable claims and utility maximization; Section~5 gives an extension of the definitions and results to the vector valued market models. In the appendix, we will collect some definitions and ideas about measurable correspondences.

\subsection{A word about notation}
On the set $\R^n$ we will denote the scalar product with $\langle\cdot,\cdot\rangle$. The corresponding norm will be denoted by $\|\cdot\|$. Relation on $\R^n$ will be induced with a closed, convex cone $K$; for $x,\,y\in\R^n$ we write $x\succeq_K y$ to mean $x-y\in K$.

The usual spaces of random variables and vectors are denoted by $L^0(\Omega,\FF,\P;\R^n)$ and $L^\infty(\Omega,\FF,P;\R^n)$. The norm on the latter space we denote by $\|\cdot\|_\infty$, which is defined $\|X\|_\infty = \esssup\|X\|$ for any $X\in L^0(\Omega,\FF,\P;\R^n)$. This notation will be abreviated in the text where no confusion can arise. The same notation will be used also for the set of measurable selections of a corresponcence $K:\Omega\rightrightarrows\R^n$, i.e. we will denote it by $L^0(\Omega,\FF,P;K)$.

Stochastic processes are time indexed collections of random variables $(X_t)_{t=0,\ldots,T}$ for some time horizon $T$. The increment of the will be denoted by $\Delta X_i = X_i-X_{i-1}$. We will sometimes need to consider stochastic processes as random vectors, i.e. elements of $L^0$ or $L^\infty$. We can then talk about their $L^\infty$ norm.

\section{Model of the financial market}
We will consider models of financial markets in finite discrete time with time horizon $T>0$. A trader can rebalance his or her position only at a finite number of time instances $\{0,1,\ldots,T-1\}$. We assume that the portfolio of the trader can be described with a vector in a finite dimensional vector space $\R^d$.

On a probability space $(\Omega,\FF,\P)$ the information available to the investor is given by a filtration $\F=(\FF_t)_{t=0,\ldots,T}$, i.e. an increasing sequence of sub-sigma algebras of $\FF$. We denote a sequence of adapted strategies with respect to the filtration $\F$ by
  $$
    \AA =\big\{(\vartheta_t)\ \big|\ 
      \vartheta_t\in L^0(\Omega,\FF_t,\P;\R^d)\,\,\,\forall t=0,\ldots,T-1\big\}.
  $$
Elements of the set $\AA$ will be called strategies; those are the actions available to investors. 

In financial mathematics one usually works with predictable strategies as there are technical reasons for that in continuous time. In the discrete time setup this distinction is irrelevant, as the classes of processes are equivalent up to re-indexation. We, thus, work with adapted strategies, as this makes notation more transparent.

In models of financial mathematics one describes the portfolio of the trader by specifying his or her holdings in each of the assets. So, if there are $d$ stocks in the market, it is sufficient to specify the number of shares of each asset that the trader is holding. 

\begin{definition}\label{def:market model}
A market model, which we will also call gains from trading, is a mapping 
\begin{align*}
  \widehat V:\AA \rightarrow L^0(\FF;\R\cup\{-\infty\})
\end{align*}
satisfying the following two axioms:
\begin{enumerate} 
\item [{\textbf{A1:}}] (normalization) 
if the agent does not participate in the market, then his or her gains from trading are zero, i.e.
  \begin{align*}
    \widehat V(0) = 0;
  \end{align*}
\item [{\textbf{A2:}}] (locality)
the gains from trading $\widehat V$ depend only locally on the strategy, i.e. for any $t$, any set $A\in \FF_t$, and strategy $\vartheta \in \AA$, we have the following:
  \begin{align*}
    \widehat V(\vartheta_0,\ldots,\vartheta_{T-1}) \indic_A 
  = \widehat V(\vartheta_0,\ldots,\vartheta_{t-1},\vartheta_t\indic_A,\vartheta_{t+1},\ldots,\vartheta_{T-1})\indic_A \quad a.s.
  \end{align*}
\end{enumerate} 
\end{definition}

Before proceeding we first briefly comment on the definition above. First note that the value $-\infty$ is allowed for the market model to indicate that the strategy is not feasible; it is effectively there in order to describe constraints in the model. 

Axiom {\textbf{A1}} is there to say that there is at least one strategy that is feasible in the market, i.e. that gives finite gains from trading. If $\widehat V:\AA \rightarrow L^0(\FF;\R\cup\{-\infty\})$ would be a mapping satisfying {\textbf{A2}} but not {\textbf{A1}} and if $\vartheta\in\AA$ would be a feasible strategy, then the translated model $\bar V(z) = \widehat V(z + \hat z) - \widehat V(\hat z)$ is a market model by Definition~\ref{def:market model}.

Axiom {\textbf{A2}} is an axiom that is always implicit in the definition of the market model. One should compare it with the fork-convexity condition of, e.g., \v{Z}itkovi\'c~\cite{vzitkovic2002filtered}. The intuition is that the gains from trading should depend only on the sequence of positions taken by the strategy and the realization $\omega$.

\begin{remark}
  Compare these axioms with the axioms for a conditional risk measure. The fundamental difference is in the measurability of the domain and codomain of the mapping. Let $\GG\subset\FF$ be a sub-sigma algebra. Then the conditional risk measure maps $\rho:L^0(\FF)\rightarrow L^0(\GG)$, where the locality property is $\rho(\indic_A X) = \indic_A\rho(X)$ for each set $A\in\GG$, i.e. measurable with respect to the codomain sigma algebra.
\end{remark}

\begin{remark}
  The choice of the codomain of the market model is somewhat arbitrary. We could have also taken as a domain any space $L^0(\Omega,\FF,P;\XX)$ where $\XX$ is some topological space. We will see that in order to define the concepts used below, we only need a topology on the space $\XX$ and a partial relation $\succeq$. The main advantage of choosing $\R$ is that one has a canonical topology and a canonical complete binary relation $\geq$ on it. Also, on this space there is a canonical minimal element $-\infty$ which we append to the space $\R$ in order to model constraints. In Section~5 we will show how the theory transfers to the codomain $\R^n$.
\end{remark}

We now provide a few canonical examples of a market model. This will, hopefully, elucidate the argument a little.

\begin{example}\label{ex:frictionless}
The basic model of a financial market is the frictionless market. The market model is described completely by an adapted process $(S_t)_{t=0,\ldots,T}$, with $S_t$ representing the price of the stock at time $t$. We tacitly assume existence of an additional bank account with zero interest rate and value 1. The final gains from trading, using a self-financing strategy, is
\begin{align*}
  \widehat V(\vartheta)
    = \sum_{t=0}^{T-1} 
       \left\langle \vartheta_t, S_{t+1} - S_t\right\rangle.
\end{align*}
Note that $\langle\cdot,\cdot\rangle$ denotes the standard scalar product on $\R^d$, which is in this case, with $d=1$, simply a product. For the market model $\widehat V$ to be 
$[-\infty,\infty)$ valued, the stock price process needs also to have values in $(-\infty,\infty)$. Then also every strategy $\vartheta\in\AA$ attains a finite value $\widehat V(\vartheta)$. 

It is easy to see how to extend the model to more then one stock. Further extension would be to add static options $\{f_i\}_{i=1}^n$ that can be purchased at initial time and held in the portfolio until the end of the trading period. It is easy to see how to extend the model to accommodate this extension.
\end{example}

\begin{example}\label{ex:additive}
To the basic frictionless market model, described by the stock price process $(S_t)_{t=0,\ldots,T}$ we add transaction costs. The capital gains are still described as in the previous example, however for every change from $\vartheta_{t-1}$ to $\vartheta_t$ of number of stocks in the portfolio, the trader incurs costs $g_t(\vartheta_t-\vartheta_{t-1})$. Those need to be paid from the bank account. The following is the model of~\cite{dolinsky2013duality} with one stock $(S_t)_{t=0,\ldots,T}$ and a bank account with zero interest rate. The market model is written here as
\begin{align*}
   \widehat V(\vartheta) = \sum_{t=0}^{T-1}\Big[
     \vartheta_t(S_{t+1}-S_t)
     - g_t(\vartheta_t - \vartheta_{t-1})\Big],
\end{align*}
with the convention $\vartheta_{-1}=0$. The transaction costs $g_t$ are in general mappings $g_t:\Omega\times\R\rightarrow\R_+\cup\{\infty\}$ such that $g_t(\Delta \vartheta_t)\in L^0(\FF;\R)$ for any random variable $\Delta \vartheta_t\in L^0(\FF_t;\R)$.
Here are a few prominent examples from mathematical finance:
\begin{description}
  \item[{\textbf{proportional transaction costs:}}] the costs are $g_t(\Delta \vartheta_t) = \lambda_t|\Delta \vartheta_t|$ for some sequence of random variables $\lambda_t>0$ a.s. One often encounters the choice $\lambda_t = \lambda S_t$ for some constant $\lambda>0$, where also the stock price process $S$ is non-negative.
  \item[{\textbf{fixed transaction costs:}}] a trader needs to pay a fixed fee $\lambda$ for every transaction, irrespective of its size. The costs $g_t$ one writes as follows 
$$ 
  g_t(x) = \left\{\begin{array}{ll}
     \lambda & x\not=0\\
     0       & \textrm{otherwise}.
  \end{array}\right.
$$
  \item[{\textbf{portfolio constraints:}}] the possible position in the stock $\vartheta_t$ at time $t$ is constrained to lie in some set $D_t\subseteq\R$. This gives rise to a somewhat different model of the form
\begin{align*}
   \widehat V(\vartheta) = \sum_{t=0}^{T-1}\Big[
     \vartheta_t(S_{t+1}-S_t)
     - \indic_{D_t}(\vartheta_t)\Big],
\end{align*}
with $\indic_{A}(x)$ is $0$ if $x\in A$ and $\infty$ otherwise.
\end{description}
\end{example}

\begin{remark}
  When trading in the market model incurs transaction costs, then it is not clear what a concept of portfolio value should be. There are two canonical candidates. The mark-to-market value is the value one obtains by multiplying the number of stocks with their 'mid-quote price' and summing those up. Of course, what is a mid-quote price depends on the model. The liquidated portfolio value is the value one would see on the bank account were one to change his or her holdings to the bank account, i.e. having $\vartheta_t=0$. In the mathematical finance literature one differentiates between the two in order to get a desired interpretation of the dual variables.
\end{remark}

\begin{example}
  Here we give a more involved example of a limit order book. More information about modelling considerations can be found in~\cite{roch2011resilient}. The 'equilibrium stock price' process $(S_t)_{t=0,\ldots,T}$ represents the price when there is no trading; with trading, after the trading period, the price is given with $S_t+\ell_t$. If a big trader wants to execute the trade $\Delta\vartheta_t$ at time $t$, he or she moves the price, so the price after trade changes by $m_t \Delta\vartheta_t$. The model is the following
\begin{align*} 
 \ell_{t+1}  &=  \kappa \ell_t + 2m_{t+1} \Delta\vartheta_{t+1}  \\
 V_{t+1} &= V_t + \vartheta_t(\Delta S_{t+1} + \kappa\Delta\ell_t) - m_t(\Delta\vartheta_t)^2
\end{align*}
for some constant $\kappa\in(0,1)$, capturing the decay of price impact $\ell$, and a strictly positive process $(m_t)$, encoding the 'depth' of the limit order book of $\frac12m_t$. We set as the market model in our $\widehat V(\vartheta) = V_T(\vartheta)$ where one can convince oneself that $V_T$ is the mark-to-market value of the portfolio $(\vartheta_t)_{t=0,\ldots,T-1}$. A particular feature of the model is that the gains from trading do not depend on the strategy in a convex manner.
\end{example}

\begin{example}
  The final example will be the one in which the market model cannot be identified with capital gains. The agent trades in the frictionless market and can also consume a part of the money. Starting with the amount $V_0 >0$ on the bank account, the agent decides on the pair $(\vartheta_t,c_t)$ where $\vartheta$ is, as above, a number of shares in the portfolio and $c_t$ is a process representing consumption; it is adapted and positive. The market model we consider in this case is
  $$
    V_T(\vartheta,c) 
      = \sum_{t=0}^{T-1}\Big[ 
         \langle \vartheta_t, S_{t+1} - S_t\rangle
         - c_{t+1}\Big].
  $$
  What one is interested in here is the utility from consumption. So, if we denote by $U:\R^T_+\rightarrow\R\cup\{-\infty\}$ the utility one achieves from the consumption stream $c$, we can define our market model as $\widehat V(\vartheta,c) = U(c)$. In order for the model to be well defined, we need to assume that the position is solvent at the end of trading, i.e. $V_T(\vartheta,c)\geq 0$ a.s.
\end{example}

One can come up with an infinite number of additional examples of market models. Let us just mention optimal stopping, which one can handle by appropriately reducing the space of strategies.


\section{Representation of market models}

It is customary in financial mathematics to think of a market as a rule assigning to the strategy a capital gains process. This dynamic view of the market through the assigned gains process mainly arises from the technicalities one needs to deal with when defining the stochastic integral in continuous time, i.e. one cannot talk about $\omega$-wise stochastic integral. 

In the examples we gave above, the market model is given in an $\omega$-wise fashion: fixing $\omega$, the market model is just a deterministic function of the path. We can write this observation as $\widehat V(\vartheta)(\omega) = V(\omega,\vartheta(\omega))$ for some function $V:\Omega\times\R^{dT}\rightarrow\R\cup\{-\infty\}$. All such mappings $V$ define market models, as the following statement shows.
  
\begin{lemma}
Let $V$ be a mapping 
  $V:\Omega\times\R^{dT}\rightarrow\R\cup\{-\infty\}$ 
and assume that it is $\FF\otimes\BB(\R^{dT})-\BB(\R\cup\{-\infty\})$ measurable. Then the mapping 
  $\widehat V:\AA\rightarrow L^0$ defined by $\widehat V(\vartheta)(\omega) := V(\omega,\vartheta(\omega))$
defines a market model when $\widehat V(0)= 0$.
\end{lemma}
\begin{proof}
  Since we assumed in the statement of the lemma that {\textbf{A1}} is satisfied, we only need to show {\textbf{A2}} and that $\widehat V$ maps into the right space. But this is clear from the assumed measurability of $V$.
\end{proof}

The purpose of this section is to show that the converse also holds under some conditions, i.e. for any $\widehat V$ we may find a mapping $V$, such that they are related as in the lemma above. 

Starting with a market model $\widehat V$, a natural thing to do would be do assign to it a mapping $V$ defined as follows
\begin{equation}\label{definitionV}
\begin{aligned}
  V:\Omega\times\R^{dT} &\rightarrow \R\cup\{-\infty\}\\
             (\omega,x) &\mapsto \widehat V(x)(\omega).
\end{aligned}
\end{equation}
Constants in $\R^{dT}$ represent deterministic, i.e. nonrandom, strategies. Those are, of course, adapted. Therefore, the definition makes sense. Hence, by Axiom {\textbf{A2}} we have also that $V(\omega,\vartheta(\omega))=\widehat{V}(\vartheta)(\omega)$ for any simple strategy, i.e. the one which takes only a finite number of values.
 However, it is not clear why the identification above would be valid beyond this class of strategies.

 One of the obstacles is that the mapping $V$ thus defined need not be measurable or have any regularity properties at all. In fact, it is easy to come up with such pathological examples; see the discussion after Proposition 14.39 in~\cite{rockafellar1998variational}. The best one could hope for is that we can modify the mapping $V$ to get a regular enough version.

We proceed in the spirit of Proposition 14.40 in~\cite{rockafellar1998variational} and define the following set of random variables
$$
  p_{x,r} = \esssup\big\{\widehat V(\vartheta)\,\big|\,
    \vartheta\in\AA, \,\|\vartheta-x\|_\infty< r\big\}  \qquad (x,r)\in \R^{dT}\times\R_+.
$$
Remember that the set $\AA$ contains all adapted processes, hence the family over which we are taking the essential supremum is not empty.
By definition of $\widehat V$ and of the essential supremum, these random variables are $\FF$ measurable. The following implication is immediate
$$
  s\geq r + \|x-y\|\quad\Longrightarrow\quad p_{x,r}\leq p_{y,s}\,\,a.s. 
$$
The condition on the left means that the ball of radius $r$ with center in $x$ is included in a ball of radius $s$ centered in $y$.

Using this preparation, we now define the candidate for the map $V$
\begin{align}\label{sec2:normal}
    V:\Omega\times\R^{dT}&\longrightarrow\R\cup\{-\infty\}\nonumber\\
     (\omega,x)          &\mapsto   \inf\big\{p_{q,r}(\omega)\, \big|\, \|x-q\|<r,\,\,q\in\Q^{dT},\,r\in\Q\cap(0,1)\big\}.
\end{align}
Note that the set over which we are taking the infimum is countable, hence it makes sense to define it in an $\omega$-wise fashion. 

As the next step, we show that this definition is good. But before going further, let us introduce some definitions. As the concept of a normal integrand is fairly standard, we chose not to invent new terms, but rather to direct the reader to~\cite{rockafellar1998variational}, Chapter 14, where the theory is expounded in great detail. Here we recall only the definition, with a warning that we change the sign in everything that follows as compared to~\cite{rockafellar1998variational}. We are interested in maximization and the book~\cite{rockafellar1998variational} deals with minimization.

\begin{definition}\label{def:normal}
A mapping $V:\Omega\times\R^{dT}\rightarrow \R\cup\{-\infty\}$ is called an $\FF$ normal integrand if its hypographical correspondence
  $$
    \omega \mapsto \hypo V(\omega) = 
      \big\{(x,r)\in\R^{dT}\times\R\ \big|\ r\leq V(\omega,x)\big\}
  $$
is closed valued and measurable with respect to $\FF$. A set valued mapping, or a correspondence, $\omega\mapsto C(\omega)$ is measurable with respect to $\FF$ if for every open ball $B\subset\R^{dT}$, we have 
$$
  \{\omega\,|\,B\cap C(\omega)\not=\varnothing\}\in\FF.
$$
Correspondences $\omega\mapsto C(\omega)$ we also denote with $C:\Omega\rightrightarrows\R^n$ to indicate that each $C(\omega)$ is a subset of $\R^n$.
\end{definition}

An overview of the standard results from the theory of measurable correspondences can be found in Appendix~\ref{ap:prvi}.

In the remainder of the chapter, we will usually drop the reference to the sigma algebra $\FF$.

A well-known representative of this class of normal integrands is the class of Carath\'eo\-dory mappings; a mapping $V:\Omega\times\R^{dT}\rightarrow \R\cup\{-\infty\}$ is Carath\'eodory if the map $x\mapsto V(\omega,x)$ is continuous for each fixed $\omega\in\Omega$ and $\omega\mapsto V(\omega,x)$ is $\FF$ measurable for each fixed $x\in\R^{dT}$; see Example~14.29 in~\cite{rockafellar1998variational}. Let us also mention that almost all the models of financial markets, with the exception of models with portfolio constraints, are represented by Carath\'eodory integrands.

The basic consequence of the definition of the normal integrand $V$ is that for all $\FF$ measurable random vectors $x\in L^0(\FF;\R^{dT})$ the random variable $\omega\mapsto V(\omega,x(\omega))$ is also $\FF$ measurable; see Proposition 14.28 in~\cite{rockafellar1998variational}.

\begin{lemma}\label{lem:normal}
  The mapping $V$, defined in~\eqref{sec2:normal} is a normal integrand.
\end{lemma}
\begin{proof}
  Let us write and expand the definition of the hypograph of the mapping $V$. We have 
\begin{align*}
  \hypo V(\omega)
&= \left\{(x,\beta)\in\R^{dT}\times\R\,|\,\beta\leq V(\omega,x)\right\}\\
&= \left\{(x,\beta)\in\R^{dT}\times\R\,\Big|\,
                                      \beta\leq p_{q,r}(\omega);\,\,
                                      \begin{array}{l}\forall q\in\Q^{dT},\ r\in\Q\cap(0,1)\\
                                      \textrm{with } \|x-q\|< r
                                      \end{array}
                    \right\}\\
&= \bigcap_{\substack{q\in\Q^{dT}\\r\in\Q\cap(0,1)}}\{(x,\beta)\in\R^{dT}\times\R\,|\,\beta\leq p_{q,r}(\omega)\,
                                                    \textrm{ if } \|x-q\|< r\}\\
&= \bigcap_{\substack{q\in\Q^{dT}\\r\in\Q\cap(0,1)}}
      \Big(
         B_r(q)\times(-\infty,p_{q,r}(\omega)]\cup B_r(q)^c\times\R
      \Big)\\
&= \bigcap_{\substack{q\in\Q^{dT}\\r\in\Q\cap(0,1)}}
      \Big(
         \R^{dT}\times(-\infty,p_{q,r}(\omega)]\cup B_r(q)^c\times\R
      \Big),
\end{align*}
where $B_r(q)$ denotes an open ball around $q$ with radius $r$. Observe first that for a fixed $\omega$ the last expression is an intersection of closed sets, therefore closed. This implies the upper-semicontinuity of the mapping $V$ for each $\omega$. When we show that each $B_r(q)\times(-\infty,p_{q,r}]\cup B_r(q)^c\times\R$ is a measurable correspondence, we will be done, by Proposition 14.11 in~\cite{rockafellar1998variational}. By the same proposition, also a finite union of measurable correspondences is measurable, hence we need to check only that the assignment $\omega\mapsto\R^{dT}\times(-\infty,p_{q,r}(\omega)]$ is a measurable correspondence.

So, let $B = B_s(y,\alpha)$ be an open ball. Then we have 
  $$
    \big\{ \omega\in\Omega\ \big|\ \R^{dT}\times(-\infty,p_{q,r}(\omega)]\cap B\not=\varnothing\big\}
      = \big\{\omega\in\Omega\ \big|\ |p_{q,r}(\omega)-\alpha|<s\big\},
  $$
which is a measurable set, by the $\FF$ measurability of $p_{q,r}$.
\end{proof}

\begin{lemma}\label{lem:bigger estimate}
  Let $\widehat V$ be a market model. Then for the normal integrand $V$ constructed above, we have
$$
  \widehat V(\vartheta)(\cdot) \leq V(\cdot,\vartheta(\cdot))\ \ a.s. \qquad\textrm{for all }\,\vartheta\in\AA.
$$
\end{lemma}
\begin{proof}
  The statement of the lemma is equivalent to the statement that for every $\vartheta\in\AA$, the random variable $(\vartheta,\widehat V(\vartheta))$ is a measurable selection of the correspondence $\hypo V$. In the proof of Lemma~\ref{lem:normal} we represented this correspondence as an intersection of correspondences 
$$
  \omega\mapsto \R^{dT}\times(-\infty,p_{q,r}(\omega)]\cup B_r(q)^c\times\R
$$
over $(q,r)\in\Q^{dT}\times\Q$ with $r>0$. 

So, fix $q$ and $r$. We need to show the following: on $A_{q,r}:=\{\omega\ |\ |\vartheta(\omega)-q|\leq r\}$ we have $\widehat V(\vartheta)\indic_{A_{q,r}} \leq p_{q,r}\indic_{A_{q,r}}$. Define the following strategy
$$
  \vartheta_{q,r} = \left(\ldots, \vartheta_t\indic_{|\vartheta_t - q_t|\leq r} + q_t\indic_{|\vartheta_t - q_t|> r} ,\ldots \right),
$$
where we have conveniently denoted $q = (q_0,\ldots,q_{T-1})$ to represent the time-indexed deterministic strategy $q$. Note now that by construction we have that $\vartheta_{q,r}\in\AA$ and also $|\vartheta_{q,r}-q|\leq r$ a.s. and $\vartheta_{q,r} = \vartheta$ on $A_{q,r}$. This finishes the proof, since by Axiom {\textbf{A2}} we have
$$
   \widehat V(\vartheta)\indic_{A_{q,r}}
   = \widehat V(\vartheta_{q,r})\indic_{A_{q,r}}
  \leq p_{q,r}\indic_{A_{q,r}}.
$$
The last inequality holds, since the strategy $\vartheta_{q,r}$ is in the set over which we take the essential supremum in the definition of $p_{q,r}$.
\end{proof}

\begin{definition}
  We say that the market model $\widehat V$ is upper semicontinuous if for every sequence of strategies $\vartheta_n\in\AA\cap L^\infty$ converging in $L^\infty$ to a strategy $\vartheta\in\AA$, we have that $\limsup_{n\rightarrow\infty} \widehat V(\vartheta_n)\leq \widehat V(\vartheta)$.
\end{definition}

The following is the main result of this section. 
\begin{theorem}\label{thm:rep1D}
  Let $\widehat V$ be the market model. If $\widehat V$ is upper-semicontinuous, then there exists a normal integrand $V$ with $\widehat V(\vartheta)(\omega) = V(\omega,\vartheta(\omega))$ for every $\vartheta\in\AA$.
\end{theorem}

The normal integrand $V$ was constructed above and Lemma~\ref{lem:bigger estimate} shows that it is bigger then the functional $\widehat V$. So, it remains to show that upper-semicontinuity of $\widehat V$ implies $\widehat V(\vartheta)(\omega) \geq V(\omega,\vartheta(\omega))$. The proof revolves around Lemma 1.3 in~\cite{peskir2006optimal}, which states that for
$$
  p_{q,r} = \esssup\big\{ \widehat V(\vartheta) \,\big|\,\vartheta\in\AA,\,\, \|\vartheta-q\|_\infty\leq r\big\}
$$
there exists a countable subset $\{\vartheta^{q,r}_k\}_{k\in\N}$ with $\|\vartheta^{q,r}_k-q\|_\infty\leq r$ for each $k$ and 
\begin{align}\label{eq:defining}
  p_{q,r} = \sup_{k\in\N} \widehat V(\vartheta^{q,r}_k).
\end{align}
We use these sequences and argue by contradiction. The proof follows.
\medskip

\begin{proof}
Let us first show that it is enough to prove the theorem for bounded strategies. So, assume that the theorem is true for all bounded strategies and let $\vartheta\in\AA$ be arbitrary. To the strategy $\vartheta$ we assign a sequence of strategies
$$
  \vartheta^n = \left(
    \vartheta_0\indic_{|\vartheta_0|\leq n},\ldots,
    \vartheta_{T-1}\indic_{|\vartheta_{T-1}|\leq n}\right).
$$
Obviously, $\vartheta^n\in\AA\cap L^\infty$ and also $\P[\vartheta^n=\vartheta]\rightarrow 1$. Denoting $A^n = \{\vartheta^n=\vartheta\}$, by the locality axiom {\textbf{A2}} we have
$$
  \widehat V(\vartheta)\indic_{A^n} 
    = \widehat V(\vartheta\indic_{A^n})\indic_{A^n} 
    = \widehat V(\vartheta^n\indic_{A^n})\indic_{A^n} 
    = V(\vartheta^n)\indic_{A^n} 
    = V(\vartheta)\indic_{A^n}.
$$
Sending $n\rightarrow\infty$ we get the equality also for $\vartheta$. Henceforth we argue with strategies bounded in $L^\infty$.
\medskip

If the statement of the theorem is not satisfied, there exists a $\vartheta\in\AA\cap L^\infty$ such that $\P[V(\cdot,\vartheta)>\widehat V(\vartheta) + 2\varepsilon]>2\varepsilon$ for some $\varepsilon>0$. We will first outline the argument in the case when the strategy $\vartheta$ is deterministic and then in Step~2 extend this argument to the general case.
\medskip

{\textbf{Step 1:}} For every deterministic strategy, which is represented simply with a vector $x\in\R^{dT}$, we claim that the theorem holds true.\medskip

Argue by contradiction. Let $(q_n)$ be a sequence in $\Q^{dT}$, satisfying $|q_n-x|<2^{-n-1}$. As noted above, since $p_{x,r}\leq p_{y,s}$ a.s. whenever $s\geq r + \|x-y\|$, it follows that 
$$
  V(\omega,x) = \lim_{n\rightarrow\infty}p_{q_n,2^{-n}}
$$
and the sequence in the limit above is decreasing.
So, we assume that the theorem is false, i.e. 
$\P[V(\cdot,x)>\widehat V(x) + 2\varepsilon]>2\varepsilon$, thus also $\P\big[p_{q_n,2^{-n}}>\widehat V(x) + 2\varepsilon\big]>2\varepsilon$ for each $n$.
By the comment we made above the proof, for each $n$ there exists a constant $k_n$ such that 
  $\P\big[\sup_{k<k_n}\widehat V(\vartheta^{q_n,2^{-n}}_k) > \widehat V(x) + \varepsilon\big]>\varepsilon$.
Denote by $(\vartheta_n)$ the sequence 
$$
\vartheta^{q_1,2^{-1}}_1,\ldots,\vartheta^{q_1,2^{-1}}_{k_1},\vartheta^{q_2,2^{-2}}_1,\ldots,\vartheta^{q_2,2^{-2}}_{k_2},\vartheta^{q_3,2^{-3}}_1,\ldots,
$$ 
By construction the sequence $\vartheta_n$ converges to $x$ in $L^\infty$. However, also by construction of the sequence, we have
$$
  \P\left[\limsup_{n\rightarrow\infty}\widehat V(\vartheta_n) \geq \widehat V(x)+\varepsilon\right]\geq\varepsilon.
$$
This contradicts the upper-semicontinuity of the market model $\widehat V$, so the assumption that $V(\omega,x) \not= \widehat V(x)(\omega)$ with positive probability was wrong.
\medskip

{\textbf{Step 2:}} The theorem holds true for every bounded strategy $\vartheta\in\AA\cap L^\infty$.\medskip

For this general case, we will follow the same strategy as in the previous step. Assuming 
$\P[V(\cdot,\vartheta)>\widehat V(\vartheta) + 2\varepsilon]>2\varepsilon$ 
for some $\varepsilon>0$ we need to construct a sequence $(\vartheta_m)$ converging to $\vartheta$ in $L^\infty$ and $\P\big[\limsup_m\widehat V(\vartheta_m) \geq \widehat V(x)+\varepsilon\big]\geq\varepsilon.$ We construct this sequence in segments, as we did in Step~1. The proof proceeds by successively (1) approximating $\vartheta$ in $L^\infty$ with a simple strategy $\theta$ with values in $\Q^d$ and $\|\vartheta-\theta\|_\infty\leq 2^{-n}$; and (2) approximating the value of $V(\theta)$ from above using defining sequences for $p_{q,r}$, given in equation~\eqref{eq:defining}. We show only how to define a finite sequence $(\vartheta_i^n)$, such that $\|\vartheta_i^n-\vartheta\|_\infty\leq 2^{2-n}$ and $\P\big[\sup_{i<k_n}\widehat V(\vartheta^n_i) > \widehat V(\theta) + \varepsilon\big]>\varepsilon$. The proof is then finished by appending those finite segments as in Step~1. The precise notation used will be laid out in the proof.

Fix an $n\in\N$. We will first construct a strategy $\theta$ such that $\|\vartheta - \theta\|_\infty<2^{-n}$. The construction is straightforward and proceeds by approximating the strategy in succession:  approximate first $\vartheta_0$ with a simple random variable
$$
 \theta_0 = \sum_{i=1}^{N_0}q(i)\indic_{\Omega(i)}
$$
such that $\|\vartheta_0 - \theta_0\|_\infty<2^{-n}$ a.s. Of course, we need $\Omega(i)\in\FF_0$. It is assumed that the sets $\Omega(i)$ are disjoint and that $\bigcup_{i=1}^{N_0}\Omega(i) = \Omega$. Proceed to the next time instance and on each $\Omega(j)$ from the previous step approximate $\vartheta_1$ with a simple random variable 
$$
\theta_1\indic_{\Omega(j)} 
 = \sum_{i=1}^{N_1(j)}q(j,i)\indic_{\Omega(j,i)}
$$
such that, after doing this for all $j$ we have $\|\vartheta_1 - \theta_1\|_\infty\leq 2^{-n}.$ It is again assumed that $\Omega(j,i)\in\FF_1$ are disjoint and $\bigcup_{i=1}^{N_1(j)}\Omega(j,i) = \Omega(j)$. We proceed in an analogous manner until we get the approximation of the last step.
Then we write 
$$
  \theta 
  = \sum_{i_0=1}^{N_0}\cdots
    \sum_{i_{T-1}=1}^{N_{T-1}(i_0,\ldots,i_{T-2})} 
    \big(q(i_0),\ldots,
    q(i_0,\ldots,i_{T-1})\big)
    \indic_{\Omega(i_0,\ldots,i_{T-1})}.
$$ 
Of course, by construction we have $\|\vartheta - \theta\|_\infty\leq 2^{-n}$. In order to further simplify the notation, we denote by $\nu$ a tuple $\nu := (i_0,\ldots,i_{T-1})$ representing a summand in the above expression and denote by $\NN$ the collection of all such indices, i.e.
$$
  \NN = \{(i_0,\ldots,i_{T-1})\subset\N^T\ |\ 
     i_0\leq N_0,\ldots, i_{T-1}\leq N_{T-1}(i_0,\ldots,i_{T-2})\}.
$$
For each $\nu\in\NN$ write the corresponding vector as $q^\nu = \big(q(\nu_0),\ldots,q(\nu_0,\ldots,\nu_{T-1})\big)$. 

Now, by the observation before the proof, for each $\nu\in\NN$ there exists a sequence $\{\vartheta_k^\nu\}_{k\in\N}\subset\AA,$ such that $\|\vartheta_k^\nu-q^\nu\|_\infty\leq 2^{1-n}$ and
$$
  p^\nu := p_{q^\nu,2^{-n+1}} = \sup_{k\in\N}\widehat V(\vartheta_k^\nu).
$$
As in the proof of Lemma~\ref{lem:bigger estimate} we now notice that
$$
  \widehat V(\vartheta)\leq V(\vartheta)\leq \sum_{\nu\in\NN}
    p^\nu\indic_{\Omega(\nu)}\qquad \textrm{a.s}.
$$
Therefore, we may also find a constant $K$ such that
$$
  \P\left[\sup_{k(\nu)\in\{1,\ldots,K\}^{|\NN|}}\sum_{\nu\in\NN}
    \widehat V\left(\vartheta^\nu_{k(\nu)}\right)\indic_{\Omega(\nu)} >\widehat V(\vartheta)+\varepsilon \right]>\varepsilon,
$$
where we look at $k(\nu)\in\{1,\ldots,K\}^{|\NN|}$ as a mapping $k:\NN\rightarrow\{1,\ldots,K\}$.

Now, if the random vectors $\sum_\nu\vartheta^\nu_{k(\nu)}\indic_{\Omega(\nu)}$ were actually adapted for each $k(\nu)\in\{1,\ldots,K\}^{|\NN|}$, we would be done. Indeed, this set of strategies is finite, hence we can just append those to our sought for sequence $(\vartheta_m)$ as in Step~1. Then, repeating the procedure for each $n$ we could conclude the proof as in Step~1.

In this last part of the proof we will show that the same can be achieved with adapted strategies. In essence, we show how to extend a strategy $\vartheta^\nu_{k}\indic_{\Omega(\nu)}$ to a strategy $\varphi$ that is adapted, an approximation for $\vartheta$, i.e. $\|\varphi - \vartheta\|_\infty\leq 2^{-n+1}$, and such that $\widehat V(\vartheta^\nu_k)\indic_{\Omega(\nu)} = \widehat V(\varphi)\indic_{\Omega(\nu)}$. Fix a $\nu\in\NN$ and a $k\in\{1,\ldots,K\}$. One can readily see that the strategy
$$
  \varphi = \left(\ldots,
    \vartheta^\nu_{k,t}
    \indic_{\Omega(\nu_1,\ldots,\nu_t)} 
    + \theta_t\indic_{\Omega(\nu_1,\ldots,\nu_t)^c}  ,\ldots\right)
$$
satisfies all the above criteria. Indeed, the last point follows from the observation that $\Omega(\nu) = \bigcap_{t=0}^{T-1}\Omega(\nu_0,\ldots,\nu_{t+1})$.
\end{proof}

\begin{remark}
  Note that in the proof we did not go on to prove that the normal integrand is finite-valued. But the proof of the theorem above shows that this is indeed the case. Assume that for some constant $q\in\Q^{dT}$ and all $r\in\Q\cap(0,1)$ we have that $\P[p_{q,r}=\infty]>\varepsilon$. Then, by the same argument as above, we could construct a sequence $\vartheta_n\rightarrow q$ uniformly, with
   $\P[\limsup_{n\rightarrow\infty}\widehat V(\vartheta_n)=\infty]\geq\varepsilon$. But this is impossible, by the assumption that $\widehat V$ takes values in 
   $L^0(\FF;\R\cup\{-\infty\})$ and the upper-semicontinuity condition.
\end{remark}

\begin{corollary}
  The semi-continuity condition of the above theorem is equivalent to the following one: for every sequence of strategies $(\vartheta^n)$, converging to a strategy $\vartheta$ a.s. we have $\limsup_{n\rightarrow\infty}\widehat V(\vartheta^n) \leq \widehat V(\vartheta)$.
\end{corollary}
\begin{proof}
  Since uniform convergence implies almost sure convergence, we only need to show that the above theorem also holds for almost sure convergence of strategies. 
  So, let $\vartheta^n$ be a sequence converging almost surely to a strategy $\vartheta$. By Egorov theorem, the convergence is almost uniform, i.e. for each $\varepsilon>0$ there exist sets $A_0\in\FF_0,\,\ldots,A_{T-1}\in\FF_{T-1}$ such that $\P[A_t]>1-\varepsilon$ for each $t$ and $\vartheta_t^n\indic_{A_t}$ converges uniformly to $\vartheta_t\indic_{A_t}$. We may also assume, that $\vartheta_t\indic_{A_t}$ is in $L^\infty$, otherwise intersect each $A_t$ with $\{|\vartheta_t|\leq N\}$ for large enough $N$. Denote now $\hat \vartheta^n = 
  (\vartheta_0^n\indic_{A_0},\ldots,\vartheta_{T-1}^n\indic_{A_{T-1}})$.
  
  From the axiom {\textbf{A2}}, we know that 
  $\widehat V(\hat \vartheta^n)\indic_{\bigcap_t A_t} = 
   \widehat V(     \vartheta^n)\indic_{\bigcap_t A_t}$. 
  Taking the limit:
  $$
    \limsup_{n\rightarrow\infty} 
         \widehat V(\vartheta^n)\indic_{\bigcap_t A_t}
  = \limsup_{n\rightarrow\infty} 
         \widehat V(\hat \vartheta^n)\indic_{\bigcap_t A_t}
  \leq \widehat V(\hat \vartheta)\indic_{\bigcap_t A_t}
  = \widehat V(\vartheta)\indic_{\bigcap_t A_t}.
  $$
  Since $\P\big[\bigcap_t A_t\big] > (1-\varepsilon)^T$ and $\varepsilon>0$ was arbitrary, the result follows.
\end{proof}

\begin{remark}
  One could now ask what is the significance of this representation result. One way to see it is to note that for any market model $\widehat V$, not necessarily upper-semicontinuous, it states that its upper-semicontinuous envelope has a particular form. Note also, that the theorem assumes sequential upper-semicontinuity of the market model $\widehat V$, which is then shown to be equivalent to semicontinuity for all $\omega$.

  Also, it allows us to extend the market model and use strategies that are not adapted, by just plugging the strategy into the function $V$. Indeed, using locality property {\textbf{A2}}, one could extend the functional $\widehat V$ beyond adapted strategies. This is true, but this extension is then possible only for 'finite pasting', i.e. one could paste a finite number of strategies. In the language of~\cite{cheridito2012conditional}, the extension is only possible to a stable hull of $\AA$. The representation theorem states that under upper-semicontinuity conditions, one can extend the functional $\widehat{V}$ to the $\sigma$ stable hull of $\AA$.
\end{remark}

\begin{example}
  The following is an example from~\cite{rockafellar1998variational} that shows that even in the case when the market model $\widehat V$ is defined as a mapping $V:\Omega\times\R^{dT}\rightarrow\R$ satisfying the property that $\widehat V(\vartheta)\in L^0(\FF;\R)$ for each strategy $\vartheta\in\AA$, the mapping $V$ might still lack the required measurability. Let $\Omega = [0,1]$ with Lebesgue measure $\P$ and Borel sigma algebra. Let $D\subset\Omega$ be a set that is not measurable. The (one step) market model is given by
  $$
    \widehat V(x)(\omega) = 
       \left\{\begin{array}{ll}
         1 & x = \omega\in D,\\
         0 & \textrm{otherwise.} 
       \end{array}\right.
  $$
  Taking the initial sigma algebra $\FF_0$ trivial, one has $\widehat V(\vartheta)=0$ a.s. for each strategy $\vartheta\in\AA.$ Notice also, that for each fixed $\omega$, the mapping $x\mapsto\widehat{V}(x)(\omega)$ is upper-semicontinuous. The problem is that the mapping $V$ is not a normal integrand; indeed, it follows immediately from the definition that the hypograph of $V$ is not measurable. However, one candidate of the construction given above would give $V(\omega,x)=0$ identically, which is measurable, i.e. a normal integrand. Also it is equal to $\widehat V$ a.s. for every $\FF_0$ measurable strategy.
\end{example}

\begin{remark}
  Semicontinuity of the market model is needed to obtain for fixed $\omega$ semicontinuity of the market representation. A natural question is whether sequential continuity of the market model implies the same for the representation. We will show by a counterexample, that the answer to this question is no in general. There is, however, a case where the result is true. This answers a question in~\cite{drapeau2013brouwer} about the connection between local (i.e. stable in their terminology), sequentially continuous maps $\widehat f:L^0(\FF;\R^d)\rightarrow L^0(\FF;\R^d)$ and Carath\'eodory integrands $f:\Omega\times\R^d\rightarrow\R^d$. In the following lemma we will consider only maps $\widehat f:L^0(\FF;\R^d)\rightarrow L^0(\FF;\R)$, which can be considered as coordinate maps. 
We say that a map $\widehat f:L^0(\FF;\R^n)\rightarrow L^0(\FF;\R)$ is sequentially continuous if for every sequence $(\vartheta^k)\subseteq L^0(\FF;\R^n)$ converging almost surely to $\vartheta\in L^0(\FF;\R^n)$, also $\widehat f(\vartheta^k)\rightarrow \widehat f(\vartheta)$ almost surely.
\end{remark}

\begin{lemma}
  Let $\widehat f:L^0(\FF;\R^d)\rightarrow L^0(\FF;\R)$ be a one step market model. If $\widehat f$ is sequentially continuous, then it can be represented by a Carath\'eodory map $f:\Omega\times\R^d\rightarrow\R$.
\end{lemma}
\begin{proof}
  Since sequential continuity implies sequential upper and lower-semi\-con\-ti\-nuity, Theorem~\ref{thm:rep1D} applies. It yields two representation integrands: $f^+$, that is upper-semicontinuous, and $f^-$, that is lower-semicontinuous. Observing the construction of those representations, we know that $f^+(\omega,\cdot)\geq f^-(\omega,\cdot)$ as functions for all $\omega$. As $f^+$ and $f^-$ are both representations for the mapping $f$, it also follows that $f^+(\omega,\vartheta(\omega))= f^-(\omega,\vartheta(\omega))$ a.s. for all $\vartheta\in L^0(\FF;\R^d)$. Consider the closed valued, $\FF$ measurable correspondences
$$
  A^\varepsilon(\omega) = \big\{(x,y)\in\R^d\times\R\,\big|\,f^-(\omega,x)+\varepsilon\leq y\leq f^+(\omega,x)\big\}.
$$
  If we can show that $A^\varepsilon=\varnothing$ almost surely for all $\varepsilon>0$, then it follows that for almost all $\omega$ the functions $f^+(\omega,\cdot)$ and $f^-(\omega,\cdot)$ coincide, i.e. outside a null-set the function $f^+(\omega,\cdot)$ is continuous, proving the claim. Assume that $\P[A^\varepsilon\not=\varnothing]>0$  for some $\varepsilon>0$. Then there exists an $\FF$ measurable selection $\vartheta$ of $A^\varepsilon$ on the set $\Omega^\varepsilon = \{\omega\in\Omega\,|\,A^\varepsilon\not=\varnothing\}$, that we set to 0 on $(\Omega^\varepsilon)^c$. But, by the representation result this implies that
$$
  \widehat f(\vartheta)(\omega) = f^-(\omega,\vartheta(\omega)) \leq f^+(\omega,\vartheta(\omega)) +\varepsilon = \widehat f(\vartheta)(\omega)+\varepsilon
$$
for $\omega\in\Omega^\varepsilon$, i.e. with positive probability. This is a contradiction, proving that $A^\varepsilon(\omega)=\varnothing$ a.s.
\end{proof}

\begin{example}
  Sequential continuity does not, in general, imply that the mapping can be represented by a Carath\'eodory map. Note, in particular, the following example. Let $\Omega=\R$ with $\GG$ trivial and a map $\widehat f:L^0(\GG;\R)\rightarrow L^0(\FF;\R)$ given with
  $$
    \widehat f(x)(\omega)=\left\{\begin{array}{ll}1&x\geq\omega\\0&\textrm{otherwise.}\end{array}\right.
    $$
 It is easy to see that the mapping $\widehat f$ is sequentially continuous with respect to any measure having no atoms. It is, however, not representable with a Carath\'eodory map.
\end{example}


\section{No-arbitrage and consequences}

The classical no-arbitrage condition in the frictionless market model of Example~\ref{ex:frictionless} is the following
  $$ 
    \sum_{t=1}^T \langle\vartheta_t,S_{t+1}-S_t\rangle\geq 0 
      \,\textrm{ a.s. }
          \quad\Longrightarrow\quad
    \sum_{t=1}^T \langle\vartheta_t,S_{t+1}-S_t\rangle = 0 
      \,\textrm{ a.s.}
  $$
The definition says that if the trading strategy does not incur losses, then it gives no gains either. There are two main reasons for considering such a definition of no-arbitrage, one technical and one philosophical. First, the definition is independent of the measure under which the condition is defined; it depends only on the null-sets of the measure $P$, and, as such, the no-arbitrage property of the market model is 'algebraic'. The second reason is that if a strategy $\vartheta\in\AA$ would be an arbitrage strategy, so would $n\vartheta$ for any $n>0$ and models having this property are not realistic, i.e. the trader cannot make his position arbitrarily large in the real world.

It is this second argument that we base the no-arbitrage condition on. The idea is that one cannot increase his or her position without increasing also the downside. To describe the market model for large values of $\vartheta$, the no-arbitrage condition is defined in terms of the recession model. It was introduced in financial mathematics framework in~\cite{pennanen2010hedging}. See also~\cite{pennanen2011convex} for a more general presentation. The recession arguments can be found also much earlier in~\cite{bertsekas1974necessary} in the context of portfolio optimization. 

In the remainder of the chapter mapping $V:\Omega\times\R^{dT}\rightarrow\R\cup\{-\infty\}$, defined in the representation theorem will be called a market model. That is, we are assuming that the market model $\widehat{V}$ is upper-semicontinuous.

\begin{definition}
  Let $V:\Omega\times\R^{dT}\rightarrow\R\cup\{-\infty\}$ be a market model. The recession market model 
  $V^\infty:\Omega\times\R^{dT}\rightarrow
      \R\cup\{\pm\infty\}$ is defined in an $\omega$-wise manner as
\begin{align*}
  V^\infty(\omega,z) = \lim_{\lambda\rightarrow\infty}\,\, \sup_{\substack{\delta>\lambda,\\|x-z|<\frac1\lambda}}\, \frac1\delta V(\omega,\delta x).
\end{align*}
\end{definition}

The mapping $V^\infty$ is well defined, as the limit in the expression above is of a decreasing sequence. By Exercise 14.54 in~\cite{rockafellar1998variational}, the mapping $V^\infty$ is also 
$\FF\otimes\BB(\R^{dT})$ measurable, such that the function $x\mapsto V^\infty(\omega,x)$ is upper-semicontinuous and positively homogeneous for all $\omega$. Note that when $V$ is a positively homogeneous market model, then $V = V^\infty$, i.e. the recession mapping is just an upper-semicontinuous regularization of the market model $V$.

\begin{remark}
  There are two reasons for restricting our attention to the upper-semicontinuous market models $\widehat{V}$ from now on. The first is that we want our market model to be upper-semicontinuous in order to prove closure of the set of the superhedgeable claims. There is, however, a more subtle reason. First, let us note that one can define the recession market model also in the general case as
$$
  \widehat V^\infty (\vartheta) = \lim_{\lambda\rightarrow\infty} \esssup_{\substack{\delta>\lambda,\\ \|\varphi-\vartheta\|_\infty<\frac1\lambda}}\, \frac1\delta \widehat V(\delta\varphi).
$$
This can be seen as a recession map of the upper-semicontinuous envelope of $\widehat V$. So, here is the subtle point: it might happen that the upper-semicontinuous envelope is not $\R\cup\{-\infty\}$ valued, i.e. takes the value $+\infty$ for some $\vartheta\in\AA$ with positive probability. By assuming upper-semicontinuity, we exclude this situation.
\end{remark}

We now define the no-arbitrage condition. The idea should be compared to conditions of the Theorem 3.10 in~\cite{rockafellar1998variational}.

\begin{definition}
  A market model satisfies the no-arbitrage ($\NA$) condition, if 
  $$
    \big\{\vartheta\in\AA\,\big|\,V^\infty(\vartheta)\geq0\,\, a.s.\big\} = \{0\}.
  $$
\end{definition}

As a simple example, one can check that the frictionless market model satisfies the $\NA$ condition exactly when it satisfies the classical no-arbitrage condition and there are no redundant assets. 

  A slightly more general market model is the one of Example~\ref{ex:additive}. One can easily check that the recession market model is given with 
  $$
    V^\infty(\vartheta) = \sum_{t=0}^{T-1}\Big[
      \vartheta_t(S_{t+1}-S_t)
     +(- g_t)^\infty(\vartheta_t - \vartheta_{t-1})\Big].
  $$ 
So, this market model satisfies the $\NA$ condition in particular when $(S_t)$ is an arbitrage free price process and $g^\infty_t(x)>0$ a.s. for all $x\not=0$; such a price process $(S_t)$ is called a strictly consistent price system in the theory of markets with proportional transaction costs.

\begin{remark}
  This condition of no-arbitrage appears already in~\cite{pennanen2011arbitrage} under the name of \emph{no scalable arbitrage} and is defined for convex transaction costs and additive structure of the market. A similar idea of the no-arbitrage condition can be found in~\cite{bouchard2011no} under the name of \emph{no marginal arbitrage of the second kind for high production regimes}. They assume that the market is arbitrage free if it can be dominated by an affine market model, satisfying the no-arbitrage condition. This also implies that it is arbitrage free under our definition.
\end{remark}

\begin{remark}  
  Here is a trivial observation, mentioned in the previous remark: if $V_1$ and $V_2$ are two market models, and $V_1(\omega,x)\leq V_2(\omega,x)$ for all $(\omega,x)$ and if $V_2$ satisfies the no-arbitrage condition, then also $V_1$ satisfies the no-arbitrage condition. One can understand the theory of proportional transaction costs through this observation: the fundamental theorem of asset pricing with proportional transaction costs states that the market model $V_1$ with proportional transaction costs is arbitrage free if and only if one can find a dominating $V_2$ model that is frictionless and also arbitrage free. To be precise, $V_2$ needs to satisfy the efficient no-arbitrage condition and admit no redundant assets.
The following example shows that this equivalence is not true anymore in non-convex market models.
\end{remark}

\begin{example}
  In a one step model with two states
   $\Omega = \{\omega_1,\omega_2\}$ and trivial initial sigma algebra, we define a market model as follows:
   $$
     V(\omega,x) = \left\{
       \begin{array}{ll}
         |x| & : \,\omega = \omega_1\\
        -|x| & : \,\omega = \omega_2
       \end{array}
       \right.
   $$
   The market model is positively homogeneous and its recession cone is just equal to $V$. It is clear that the market model satisfies the no-arbitrage condition. However, we cannot dominate it with any frictionless market model.
\end{example}

Let $f\in L^0(\FF;\R\cup\{-\infty\})$ be a random variable. Denote by $\CC_f$ the set of all superhedgeable claims that dominate $f$, i.e.
  \begin{align}
    \CC_f = \big\{g\in L^0(\FF;\R)\,\big|\, \exists z\in\AA:\,\,V(z)\geq g\geq f\,\,\textrm{a.s.}\big\}.
  \end{align}
Also, define by $\CC$ the set of all superhedgeable claims in the market, i.e. denote $\CC = \CC_{-\infty}$. 

We now define market viability. What we mean by that is the minimal assumption on the market model which makes it sensible to talk about expected utility maximization. 

\begin{definition} 
  We say that the market model is viable if the set $\CC_f$ is bounded in probability for every finite random variable $f\in L^0(\FF;\R)$.
\end{definition}

\begin{lemma}\label{lem:efficient_liability}
   Market models satisfying the $\NA$ condition are viable.
\end{lemma}
\begin{proof}
   In the proof, we consider $f\in L^0(\FF;\R)$ given and fixed. Assume that $\CC_f$ is not bounded in probability and let $(\vartheta_n)\subset \AA$ be a sequence of strategies such that $P[V(\vartheta_n)\geq n]>\delta$ for some $\delta>0$. We will show that the sequence $(\vartheta_n)$ is unbounded in probability and the rest follows by Lemma~\ref{lem:closure} of the subsequent section. 

  For almost all $\omega$, the functions $x\mapsto V(\omega,x)$ are upper-semicontinuous and $[-\infty,\infty)$ valued, hence they are bounded above on compacts. Hence, there exists a sequence of finite random variables $(m_i)$ such that $V(\omega,x)\leq m_i(\omega)$ for each $x\in B_i(0)$ in the ball with radius $i$ and centered in $0$. 
   
   Assume that the sequence of strategies $(\vartheta_n)$ is bounded in probability. By definition, for each $\varepsilon>0$ we can find an $M>0$, such that $\sup_{n\in\N}\P[|\vartheta_n|\geq M]< \frac{\varepsilon}{2}$. Next, choose an $N>0$ such that also $\P[\sup_{i=1,\ldots,M}m_i\geq N]<\frac{\varepsilon}{2}$. We estimate
\begin{align*}
   \P[|V(\vartheta_k)|>N]&=
\P[|V(\vartheta_k)|>N,\,|\vartheta_k|>M]+\P[|V(\vartheta_k)|>N,\,|\vartheta_k|\leq M]\\
&\leq \P[|\vartheta_k|>M] + \P\left[\sum_{i=1}^M m_i\indic_{|\vartheta|\in(i-1,i]}>N\right]\\
&\leq \frac{\varepsilon}{2} + \P\left[\sup_{i=1,\ldots,M}m_i >N\right]\\
&\leq \frac{\varepsilon}{2}+\frac{\varepsilon}{2} = \varepsilon.
\end{align*}
This implies that also the set of strategies $(\vartheta_n)$ is unbounded in probability. In turn, this implies, by Lemma~\ref{lem:closure}, that the market model cannot satisfy $\NA$.
\end{proof}

\begin{remark}
  The reverse implication is in general not true. The $\NA$ condition is only sufficient, but not necessary condition for viability. Already the market model $V(z)=0$, which is not arbitrage free by our definition, does not satisfy the condition, although the model is clearly viable.
\end{remark}

For a random variable $f$, define a subset $\AA_f$ of strategies that superhedge $f$, i.e.
   $$
     \AA_f = \big\{\vartheta\in\AA\,\big|\,V(\vartheta)\geq f\,\,\textrm{ a.s.}\big\}.
   $$ 

\begin{lemma}\label{lem:closure}
   Let $f\in L^0(\FF;\R)$ be a random variable and let $V$ be a market model satisfying the\, $\NA$ condition. Then $\AA_f$ is bounded in probability.
\end{lemma}

   This is essentially a variation on Theorem 3.10 of~\cite{rockafellar1998variational}; compare also to~\cite{pennanen2011convex}. The proof is based on the ideas of random subsequences: 

\begin{lemma}
  Let $(f_n)\subset L^0(\FF;\R^n)$ be a sequence of random vectors such that  $\P\left[\liminf_n|f_n|<\infty\right]=1$. Then, there exists an increasing sequence of random variables $(\tau(n))\subset L^0(\FF;\N)$ such that $(f_{\tau(n)})$ converges almost surely to a random variable in $L^0(\FF;\R)$.
\end{lemma}

\begin{proof}\emph{of Lemma~\ref{lem:closure}.}
The argument is by contradiction: assume that $\AA_f$ is not bounded in probability. We may choose a sequence $(\vartheta^k)_{k\in\N}\subset\AA_f$ that is already unbounded in probability. Then, by using the idea of random subsequences, we want to show that there exists a strategy $\hat \vartheta$ with $V^\infty(\hat \vartheta)\geq0$ a.s.  

We repeat the following algorithm for every step $t$, starting with $t=0$ and proceeding in increasing order. By $(\vartheta^k)$ we denote the original sequence above and work instead with a modified sequence $(\hat\vartheta^k)$; the first version of this sequence, i.e. before starting the algorithm, we define as $\vartheta^k = \hat\vartheta^k$. 

If the set $\{\omega\,|\,\liminf_{k\rightarrow\infty} |\hat\vartheta_t^k|=\infty\}$ has non-zero probability, denote it by $A$ and define a sequence of random variables $\lambda_t(k) = \frac{1}{1+|\hat\vartheta_t^k|}$ on $A$ and zero outside; otherwise set $\lambda_t(k) = 1$. We pass to an $\FF_t$ measurable subsequence $(\tau_t(k))$, such that, by the above Lemma, the sequence of random variables $\lambda_t(\tau_t(k))\hat\vartheta_t^{\tau_t(k)}$ converges almost surely. Before continuing with the next step $t+1$, replace the original sequence $\hat\vartheta^k$ by $\lambda_t(\tau_t(k))\hat\vartheta^{\tau_t(k)}$.

  Denoting 
$$
  \tau(k) = \tau_0(\ldots\tau_{T-2}(\tau_{T-1}(k))\ldots)
$$ 
and 
$$
  \lambda(k) = \prod_{t=0}^{T-1}\lambda_t(\tau_t(\ldots\tau_{T-2}(\tau_{T-1}(k))\ldots)),
$$ 
we may readily see that those sequences were constructed in such a way that $\hat\vartheta^k$, after finishing the procedure above, satisfies $\hat\vartheta^k = \lambda(k)\vartheta^{\tau(k)}$ and also $\hat\vartheta^k\rightarrow\varphi$ a.s. for some adapted strategy $\varphi\in\AA$. 

   Now we show that $V^\infty(\varphi)\geq0$ almost surely. We have
\begin{align*}
0 &= \lim_{k\rightarrow\infty} 
   \lambda(k)f\\
  &\leq
    \lim_{k\rightarrow\infty} 
       \lambda(k) V(\vartheta^{\tau(k)})\\ 
  &=
    \lim_{k\rightarrow\infty} 
     \lambda(k)
        V\left(\frac{\hat\vartheta^{\tau(k)}}{\lambda(k)}\right)\\
  &\leq \lim_{k\rightarrow\infty} 
     \sup_{\substack{
        \lambda\in(0,\frac{1}{\lambda(k)})\\
        |x-\varphi(\omega)|\leq \lambda(k)}} 
          \lambda V\Big(\frac{x}{\lambda}\Big) \\
  &= 
    V^\infty(\varphi).
\end{align*}
This is a contradiction, since, by construction, the strategy $\varphi$ is non-zero and we assumed that the model satisfies the no-arbitrage condition.
 Note that the above calculation is done in an $\omega$-wise fashion.
\end{proof}

\begin{lemma}\label{lem:von}
Let $(\vartheta^n)\subset\AA$ be a sequence that is bounded in probability. Then, there exists an increasing sequence of random variables $(\tau(n))\subset L^0(\FF;\N)$ such that $(\vartheta^{\tau(n)})$ converges almost surely to an adapted strategy $\vartheta\in\AA$.
\end{lemma}
\begin{proof}
We follow a similar argument as in the previous proof. Take an $\FF_0$ measurable subsequence $\tau_0(k)$, such that the sequence $\vartheta_0^{\tau_0(k)}$ converges to a finite random variable $\hat\vartheta_0$. Then pass to the $\FF_1$ measurable subsequence $\tau_1(k)$ of $\tau_0(k)$, such that $\vartheta_1^{\tau_1(k)}$ converges to $\hat\vartheta_1$. We proceed in this way, until we obtain a random subsequence $\tau(k)$, such that $\vartheta^{\tau(k)}$ converges to a predictable strategy $\hat\vartheta$.
\end{proof}

\begin{theorem}\label{thm:closure}
  If the market model satisfies the no-arbitrage condition, then the set $\CC$ of superhedgeable claims is closed in probability.
\end{theorem}
\begin{proof}
  Choose a sequence of random variables $h^k\in \CC$ that converges in probability to some random variable $h\in L^0$. By passing to a subsequence, we may assume that the sequence converges almost surely to $h$. Let $(\vartheta^k)$ be the corresponding sequence of strategies, such that $h^k\leq V(\vartheta^k)$. Now, note  that almost sure convergence implies that the pointwise minimum $f = \inf_k h^k$ is a finite random variable. By the lemma above, we know that the set $(\vartheta^k)$ is bounded in probability. 

Let $(\tau(k))\subset L^0(\FF,\N)$ be the increasing sequence, defined in Lemma~\ref{lem:von}. 
Note that we still have $h^{\tau(k)}\leq V(\vartheta^{\tau(k)})$. Noting that $V$ is pointwise upper semicontinuous, we get
$$ 
   h = \lim_{n\rightarrow\infty} h^{\tau(k)} 
    \leq \limsup_{k\rightarrow\infty} V(\vartheta^{\tau(k)})
    \leq V(\vartheta), 
$$
which shows that $f\in\CC$.
\end{proof}

\begin{remark}
  The previous statement one usually calls the 'superhedging theorem'. Let $f\in L^0(\FF;\R)$ be a random variable, i.e. the payoff of a contingent claim. We define the price of $f$ to be the minimal amount $p\in\R$ such that $f-p$ is superhedgeable, i.e.
$$
  \rho(f) = \inf\big\{p\in\R\ \big|\ f-p\in\CC\big\}.
$$
The result above states that if $\rho(f)$ is finite, then there exists a strategy $\vartheta^f\in\AA$ such that $f\leq\rho(f)+V(\vartheta^f)$ a.s. This is clear, since the sequence $f-2^{-n}\subset\CC$ converges a.s. to $f$, and by the preceeding theorem, the claim follows.
\end{remark}

We will now show that if the market model satisfies the $\NA$ condition, then there is a strong bound on the final wealth. We start with the following observation.

\begin{theorem}\label{theorem:inf-compactness}
  Let $V$ be a market model satisfying $\NA$ and $f\in L^0(\FF;\R)$ a random variable. Then there exists a finite random variable $K$ such that any strategy $\vartheta\in\AA_f$ that superhedges $f$ satisfies $|\vartheta|\leq K$ almost surely.
\end{theorem}
\begin{proof}
  By Lemma~\ref{lem:closure} the set of strategies satisfying the admissibility condition is bounded in probability. Here, we use that observation to prove a stronger bound. We will proceed by induction on the number of time steps. 
  
  To show boundedness of the first step of the strategy we define the following set
  \begin{align*}
    \Psi_0 = \big\{h\in L^0(\FF_0;\R_+)\,\big|\,
    \exists \vartheta\in \AA_f:\,|\vartheta_0|\geq h\big\},
  \end{align*}
  that contains all the norms of restrictions of the superhedging strategies of $f$ to the first component. Note that the set is bounded in probability and also upward-directed. So, there exists a sequence of strategies $(\vartheta^k)\subset\AA$ such that the following holds: $\lim_{k\rightarrow\infty}|\vartheta_0^k| = \esssup\Psi_0$. Since the set of strategies superhedging $f$ is bounded in probability, so is the set $\{|\vartheta_0^k|\,|\,k\in\N\}$, hence the random variable $m_0 = \esssup\Psi_0$ is finite.

  Next, define the reduced market model, defined $\omega$-wise, as follows
  \begin{align*}
     V_1(\omega,x_1,\ldots,x_{T-1}) =  \sup_{\substack{y\in\R^d\\|y|\leq m_1(\omega)}} V(\omega,y,x_1,\ldots,x_{T-1}).
  \end{align*}
  Note that for each $\omega$ maximization is over a compact set, which implies that $V_1(\cdot)$ is upper-semicontinuous almost surely (see Definition 1.16 and Theorem 1.17 in~\cite{rockafellar1998variational}).
$V_1$ is also jointly measurable, hence a market model (see Corollary 14.34 and Proposition 14.47 of~\cite{rockafellar1998variational}).
  
  It is also clear that 
    $V_1(\vartheta_1,\ldots,\vartheta_{T-1})\geq f$ 
  for all strategies $\vartheta\in\AA_f$ that superhedge $f$, by the construction of $V_1$.
  
  We claim that the market model $V_1$ also satisfies the no-arbitrage condition. First we calculate the recession function of the market model
  \begin{align*}
    V_1^\infty(x_1,\ldots,x_{T-1}) 
&= \lim_{\delta\searrow 0}\,\, \sup_{\substack{\lambda\in(0,\delta),\\|x-u|<\delta}}\, \lambda V_1\Big(\frac{u_1}{\lambda},\ldots,\frac{u_{T-1}}{\lambda}\Big)\\
&= 
\lim_{\delta\searrow 0}\,\, \sup_{\substack{\lambda\in(0,\delta),\\|x-u|<\delta}}\,\sup_{\substack{y\in\R^d\\|y|\leq m_1(\omega)}} \lambda V\Big(y,\frac{u_1}{\lambda},\ldots,\frac{u_{T-1}}{\lambda}\Big)\\
&= 
V^\infty(0,x_1,\ldots,x_{T-1})
  \end{align*}
  Hence, if the original market model $V$ satisfies the no-arbitrage condition, then also the reduced market model $V_1$ satisfies the condition.

   Repeating the procedure as for the first step, one gets random variables $m_0$, \ldots, $m_{T-1}$. We can then define the random variable $K$ as $K = m_0 + \cdots + m_{T-1}$.
\end{proof}

\begin{lemma}
  If the market model satisfies the $\NA$ condition, then the random variable $M_f = \esssup \CC_f$ is finite.
\end{lemma}
\begin{proof}
  Let $M_i$ be the finite random variables, such that $V(x)\leq M_i$ for each $x\in B_i(0)$. Choosing the random variable $K$ from the previous Theorem, we get the following estimate:
  \begin{align*}
    M_f\leq \sum_{i=1}^\infty M_i\indic_{K\in(i-1,i]}
  \end{align*}
which is finite.
\end{proof}

We now give a basic utility maximization result. The basic statement of existence of a maximizer basically follows from~\cite{evstigneev1976measurable}. We give a fairly general formulation of it here.

A random utility function $U$ on the positive half-line is a mapping 
$U:\Omega\times\R_+\rightarrow\R$ such that 
\begin{enumerate}
 \item for each $\omega\in\Omega$, the mapping $x\rightarrow U(\omega,x)$ is non-decreasing. 
 \item for every random variable $X\in L^0(\FF)$ we have $\omega\mapsto U(\omega,X(\omega))$ measurable with respect to the sigma algebra $\FF$.
\end{enumerate}
We will usually omit the dependence of $U$ on $\omega$. In the language of convex analysis, we may say that the mapping $U$ is an increasing normal integrand. The properties needed are such as to make the compositum $U\circ V$ a normal integrand; see Proposition~14.45 in~\cite{rockafellar1998variational}.
The basic statement in the utility maximization setup is the following

\begin{theorem}
  Let $U$ be a random utility function and $f$ a random variable. Let $M = \esssup M_0$ and assume that  $U(M)\in L^1$. Then an optimizer to the following utility maximization problem exists
  $$
    \sup_{\vartheta\in \AA_0}
       \E[U(V(\vartheta))] = 
       \E[U(V(\vartheta^\ast))].  
  $$
\end{theorem}
\begin{proof}
  We may assume that the problem is non-trivial, i.e. that there exists a strategy $\vartheta$ such that $\E[U(V(\vartheta))]>-\infty$. Otherwise any strategy is utility maximizer. 

   If the problem satisfies the no-arbitrage condition, then the optimal strategy exists by Theorem~2 of~\cite{evstigneev1976measurable}. To use the notation in the cited paper set $\varphi(x,\omega)=U(\omega,V(\omega,x))$. Then: (i) $\varphi$ is measurable by assumption; (ii) it is dominated by the random variable $U(M)\in L^1$; (iii) $\inf$-compactness follows from Theorem~\ref{theorem:inf-compactness}.
\end{proof}

Note that in the above one does not need to assume that the utility function be concave. Indeed, the proof of the cited Theorem~2 of~\cite{evstigneev1976measurable} goes through via dynamic programming. Only duality approaches require convexity. 

It was recently shown in~\cite{penannen2015non} that the $\inf$-compactness may be removed in the cited result of~\cite{evstigneev1976measurable} and replaced with the recession condition identical to our no-arbitrage condition. The uniform lower bound on the wealth is, therefore, unnecessary. 

We will give the example with fixed transaction costs, which was the main example that motivated this research.

\begin{example}
Let $S$ be the $d$-dimensional frictionless stock price process and $\lambda>0$ the costs the trader needs to pay for every rebalancing of the portfolio. For a strategy $\vartheta\in\AA$, the final wealth of trading is given with
\begin{align*}
  V(\omega,\vartheta)  
   = \sum_{t=1}^T\Big[ 
   \langle \vartheta_t, S_t(\omega) - S_{t-1}(\omega)\rangle
   - \lambda\indic_0(\vartheta_t)
  \Big].
\end{align*}
Note that the indicator function here $\indic_B(x)$ takes value 1 if $x\in B$ and 0 otherwise.

The market model given by $V$ is arbitrage free precisely when the frictionless stock price process $S$ is arbitrage free and there are no redundant assets. Indeed, observe that $V^\infty(\omega,\vartheta) = \sum_{t=1}^T \langle \vartheta_t, S_t(\omega) - S_{t-1}(\omega)\rangle$ is precisely the frictionless part of the market, from which the claim follows.
  But, one can see this also directly by noticing that in the finite discrete time frictionless market, if there exists an arbitrage opportunity, there also exists a buy-and-hold arbitrage opportunity. 

Hence we get that in the fixed transaction costs model, a utility maximizer always exists for a (random) utility function that is bounded above by an integrable random variable.
\end{example}


\section{Market models taking values in $\R^n$}

In this section we will show how to extend the definitions and results from previous sections to the case when the market model $\widehat V$ may take values in a more general space and not only $\R$. For simplicity of the argument and mathematics, we will restrict our attention to the case when
$$
  \widehat V:\AA\rightarrow L^0(\FF;\R^n\cup\{-\infty\}),
$$
where $\AA$ stands, as above, for the space of all adapted strategies. The point $-\infty$ is appended to the space; a strategy taking the value $-\infty$ with non-zero probability is deemed infeasible. 

The axioms of the market model from the Definition~\ref{def:market model} transfer here verbatim. The Axiom~{\textbf{A1}} is a normalization of the market model, stating that $\widehat V(0)=0$ a.s. As for the Axiom~{\textbf{A2}}, it transfers verbatim: for any $t\in\{0,\ldots,T-1\}$, any set $A\in \FF_t$, and strategy $\vartheta \in \AA$ the following holds
  \begin{align*}
    \widehat V(\vartheta_0,\ldots,\vartheta_{T-1}) \indic_A 
  = \widehat V(\vartheta_0,\ldots,\vartheta_{t-1},\vartheta_t\indic_A,\vartheta_{t+1},\ldots,\vartheta_{T-1})\indic_A \quad a.s.
  \end{align*}
That is to say, modifying the strategy outside the set $A\in\FF_t$ does not modify the outcome on the set $A$.

This set of models encompases at least the Kabanov model of foreign exchange markets~\cite{kabanov1999hedging}, and the adaptation of this model for the study of markets with incomplete information by~\cite{bouchard2006no}. The thinking of this latter paper is the main motivation for our approach.

\begin{remark}
In Definition~\ref{def:market model} the value $-\infty$ is the minimal value for the canonical relation on the set $\R$. In the case of $\R^n$, with a relation that is not complete, there could be more then one notion of minimal element. One notion is to proclaim an element, denoted by $-\infty$, a minimal element for the relation and append it to the space. This suffices for our purposes. Another way to think about the the minimal elements is to use the idea of cosmic closure, see Chapter~3 in~\cite{rockafellar1998variational}. To the set $\R^n$ one adds the points at infinity; mapping the set $\R^n$ with the mapping $\Phi:x\mapsto\frac{x}{1+|x|}$ to the open unit ball. Points at infinity are identified with the boundary of the image $\Phi(\R^n)$. One can thus talk about convergence of a sequence $(x_n)\subset\R^n$ to a point at infinity through the transformation of the sequence with the mapping $\Phi$. 
\end{remark}

The space $\R$ is endowed with the canonical complete relation $\geq$, hence it did not require an explicit mention. However, on the space $\R^n$ we need to define explicitly a partial order. This will be defined with a closed convex cone $K$. The order generated by the cone $K$ will be denoted by $\succeq_K$; that is, for two vectors $a,\,b\in\R^n$ we have $a\succeq_K b$ if and only if $a-b\in K$. Equivalence for the relation $\succeq_K$ is denoted by $a\sim_K b$.

Let $K:\Omega\rightrightarrows\R^n$ be a closed, convex cone-valued measurable correspondence. We assume that the cone satisfies $\mathrm{int}\, K(\omega)\not=\varnothing$ for almost all $\omega$. On the space of random vectors $L^0(\FF;\R^n)$ we define a partial relation in an almost sure sense: for vectors $X$ and $Y\in L^0(\FF;\R^n)$ we write $X\succeq_KY$ if for almost all $\omega$ we have $X(\omega)\succeq_{K(\omega)}Y(\omega)$. 

 The positive polar of the cone of $K$ is denoted by $K^\circ$ and defined by 
  $$
    K^\circ(\omega) = \big\{y\in\R^n\,\big|\,\langle x,y\rangle\geq0\,\,\,\,\forall x\in K(\omega)\big\}.
  $$
This is a set-valued map with closed, convex, and conical values. It is also an $\FF$ measurable correspondence by Exercise~14.12~(e) in~\cite{rockafellar1998variational}. The random partial order can now be expressed in the following, equivalent, way: for two random vectors $X,\,Y\in L^0(\FF;\R^n)$ we have $X\succeq_K Y$ if and only if $\langle Z,X\rangle\geq \langle Z,Y\rangle$ a.s. for each $\FF$ measurable selection of the correspondence $K^\circ$. The set of $\FF$ measurable selections of the correspondence $K^\circ$ will be denoted by $L^0(\FF;K^\circ)$. 

The relation is extended to the point $-\infty$ by setting $X\succeq_K -\infty$ for every $X\in L^0(\FF;\R^n)$. Also, we extend the scalar product by setting $\langle Z,-\infty\rangle=-\infty$ for each selection $Z\in L^0(\FF;K^\circ)$.

This, first, observation allows us to infer the order on the space $\R^n$ using a countable family of functions.
\begin{lemma}
  Let $\c=\{Z_k\,|\,k\in\N\}$ be a Castaing representation of the measurable correspondence $K^\circ$. Then for random vectors $X$, $Y\in L^0(\FF;\R^n)$ we have $X\succeq_K Y$ if and only if $\langle Z_k,X\rangle\geq\langle Z_k,Y\rangle$ almost surely for all $k\in\N$.
\end{lemma}

\begin{proof}
 This follows directly form the bipolar theorem, see Corollary~6.21 in~\cite{rockafellar1998variational}. So, for almost all $\omega$ we have
$$
  K(\omega) 
   = \big\{x\in\R^n\,\big|\,\langle x,y\rangle\,\,\forall y\in K^\circ(\omega)\big\}
   = \big\{x\in\R^n\,\big|\,\langle x,Z_k(\omega)\rangle\,\,\forall k\in\N\big\}.
$$
\end{proof}

By the same argument as in the lemma above, one can show that for random vectors $X$, $Y\in L^0(\FF;\R^n)$ we have $X\succeq_K Y$ if and only if $\langle Z,X\rangle\geq\langle Z,Y\rangle$ almost surely every measurable selection $Z$ of the correspondence $\textrm{ri}\,K^\circ$.

The definition of upper-semicontinuity for maps with values in $\R^n$ needs a new idea. As the relation induced by the random cone is not complete, the concept of supremum is not straight forward. The ideas about how to define a supremum in this setup are a topic of~\cite{kabanov2013essential}. The construction they suggest for the supremum with respect to cone induced relation yields a set, i.e. a measurable correspondence in the random setting. It is not clear whether one can obtain single valued definition of limes superior using their ideas.

In order to be able to work with the concept of upper-semicontinuity in a simple way, we will define it through scalarizations of the market model.

\begin{definition}\label{def:usc}
  A market model $\widehat V:\AA\rightarrow L^0(\FF;\R^n\cup\{-\infty\})$ is upper-semiconti-nuous if for every $\FF$ measurable selection $Z$ of the relative interior $\textnormal{ri}\,K^\circ$ the mapping 
\begin{align*}
  \widehat V_Z:\AA&\rightarrow L^0(\FF;\R\cup\{-\infty\})\\
  \vartheta&\mapsto\big\langle Z,\widehat V(\vartheta)\big\rangle
\end{align*}
is upper-semicontinuous.
\end{definition}

An immediate consequence of the definition of upper-semicontinuity is that for each selection $Z\in L^0(\FF;\textrm{ri}\,K^\circ)$ there is a representation $V_Z:\Omega\times\R^{dT}\rightarrow\R\cup\{-\infty\}$ of the market model $\widehat V_Z$. To obtain the representation theorem, we would now like to use the scalarized market models $V_Z$  to reconstruct the market model $V:\Omega\times\R^{dT}\rightarrow\R^n\cup\{-\infty\}$ that represents the market model in an appropriate way. Let us first show that if this would be possible, then $V$ would have the required measurability properties.

\begin{lemma}
  Let $\widehat V$ be an upper-semicontinuous market model with representations $V_Z$ for each $Z\in L^0(\FF;\textnormal{ri}\,K^\circ)$. If there exists a mapping $V:\Omega\times\R^{dT}\rightarrow\R^n\cup\{-\infty\}$, such that we have $V_Z = \langle Z,V\rangle$ for each $Z\in L^0(\FF;\textnormal{ri}\,K^\circ)$, then its hypograph correspondence 
$$
  \hypo V(\omega) = \big\{(x,y)\in\R^{dT}\times\R^n\,\big|\, V(\omega,x)\succeq_{K(\omega)}y\big\}
$$ 
is closed valued and $\FF$ measurable.
\end{lemma}
\begin{proof}
Let $\c=\{Z_k\,|\,k\in\N\}$ be Castaing representation for $K^\circ$ such that $\c\subset L^0(\textrm{ri}\,K^\circ)$.
Let us expand the definition of the hypograph of the mapping $V$
\begin{align*}
  \hypo V(\omega)
  &=\{(x,y)\in\R^{dT}\times\R^n\,|\, V(\omega,x)\succeq_{K(\omega)}y\}\\
  &=\{(x,y)\in\R^{dT}\times\R^n\,|\, \forall Z\in\c,\,\,V_Z(\omega,x)\geq\langle Z(\omega),y\rangle\,\}\\
  &=\bigcap_{Z\in\c}\left\{(x,y)\in\R^{dT}\times\R^n\,|\, V_Z(\omega,x)\geq\langle Z(\omega),y\rangle\,\right\}\\
  &=\bigcap_{Z\in\c}\left\{(x,y)\in\R^{dT}\times\R^n\,|\, (x,\langle Z(\omega),y\rangle)\in\hypo V_Z(\omega)\,\right\}.
\end{align*}
The last expression is an $\FF$ measurable by Example~14.15~(b) in~\cite{rockafellar1998variational}. It remains to notice that it is also closed as an intersection of closed correspondences. 
\end{proof}

The existence of such a mapping $V$ in the above Lemma is not an easy question. The reason lies in the following observation: Let $Z_1$ and $Z_2$ be two measurable selections of $L^0(\textrm{ri}\, K^\circ)$. Then the natural representation for the scalarization $\widehat V_{Z_1+Z_2}$ is the mapping $V_{Z_1}+V_{Z_2}$. To show existence of a representation $V$ for our market model $\widehat V$ one needs to show that our particular construction of the market model representation yields
$$
  V_{Z_1}(\omega,\cdot)+V_{Z_2}(\omega,\cdot)=V_{Z_1+Z_2}(\omega,\cdot)\qquad \textrm{for each }\,\omega.
$$
A careful inspection of the proofs will show that our construction of the representation in the $n=1$ case yields only $V_{Z_1}(\omega,\cdot)+V_{Z_2}(\omega,\cdot)\geq V_{Z_1+Z_2}(\omega,\cdot)$ as functions for each $\omega$. For that reason, it is, in general, not possible to find $V:\Omega\times\R^{dT}\rightarrow\R^n\cup\{-\infty\}$ such that $\langle Z,V\rangle=V_Z$ for all $Z\in L^0(\ri K^\circ)$. However, we do get a representation if we relax the requirement of upper-semicontinuity of a mapping $V$.

Let us now construct a candidate for the mapping $V$ to represent the market model.
One can easily check that the set $L^0(\FF;\textrm{ri}\,K^\circ)$ of $\FF$ measurable selections of the measurable correspondence $\textrm{ri}\,K^\circ$ satisfies the conditions of Theorem~2.8 in~\cite{cheridito2012conditional}. Therefore, there exist pairwise disjoint sets $A_1,\ldots,A_n\in\FF$ and $Z_1,\ldots,Z_n\in L^0(\FF;\textrm{ri}\,K^\circ)$ such that 
\begin{enumerate}
  \item $\bigcup_{i=1}^nA_i=\Omega$ and the linear span of $K^\circ(\omega)$ is $i$ dimensional in $\omega\in A_i$;
  \item for all $i\in\{1,\ldots,n\}$ with $\P[A_i]> 0$, $Z_1(\omega),\ldots,Z_i(\omega)$ are linearly independent on $\omega\in A_i$.
\end{enumerate}
We now construct the mapping $V$. For all $i\in\{1,\ldots,n\}$ with $\P[A_i]>0$ we write $V(\omega,x)=\sum_{k=1}^i\alpha_k(\omega,x)Z_k$ for some mappings $\alpha_k:\Omega\times\R^{dT}\rightarrow\R\cup\{-\infty\}$ that are determined with
$$
  \left[\begin{array}{ccc}
  \langle Z_1(\omega),Z_1(\omega)\rangle&\cdots&\langle Z_1(\omega),Z_i(\omega)\rangle\\
     \vdots&\ddots&\vdots\\
  \langle Z_i(\omega),Z_1(\omega)\rangle&\cdots&\langle Z_i(\omega),Z_i(\omega)\rangle\end{array}\right]
  \left[\begin{array}{c}\alpha_1(\omega,x)\\\vdots\\\alpha_i(\omega,x)\end{array}\right]=
  \left[\begin{array}{c}V_1(\omega,x)\\\vdots\\V_i(\omega,x)\end{array}\right].
$$ 
  The matrix is invertible for all $\omega\in A_i$ by linear independence of $(Z_k)$, hence $\alpha_k$ are uniquely defined, and thus also $V$. Note also, that $V$ is $\FF\otimes\BB(\R^{dT})$ measurable, since $V_i$ are and the matrix above has $\FF$ measurable elements.

  We can now easily see that the mapping $V$ represents the market model in the following sense: for each strategy $\vartheta\in\AA$ the identity $\widehat V(\vartheta)(\omega)\sim_{K(\omega)} V(\omega,\vartheta(\omega))$ holds almost surely. To see that, choose an $i\in\{1,\ldots,n\}$ with $\P[A_i]>0$. Every selection $Z\in L^0(\FF;\textrm{ri}\,K^\circ)$ can be written on a set $A_i$ as a linear combination of the basis vectors $Z_1,\ldots,Z_i$ used above to construct $V$: $Z = \beta_1Z_1 + \cdots + \beta_iZ_i$ for some weights $\beta_k\in L^0(\FF;\R)$. Therefore also
\begin{align*}
  \big\langle Z(\omega),V(\omega,\vartheta(\omega))\big\rangle
  &= \sum_{k=1}^i\beta_k(\omega)V_k(\omega,\vartheta(\omega))\\
  &= \sum_{k=1}^i\beta_k(\omega)\widehat V_k(\omega,\vartheta(\omega))
  = \widehat V_Z(\vartheta)(\omega) \qquad a.s.
\end{align*}
This implies that $V(\omega,\vartheta(\omega))\sim_{K(\omega)}\widehat V(\vartheta)(\omega)$ on the set $\omega\in A_i$. Since the same procedure holds for each $i$, it also holds almost surely.

We have shown the following.

\begin{theorem}\label{thm:Rn representation thm}
  Let $\widehat V:\AA\rightarrow L^0(\R^n\cup\{-\infty\})$ be an upper-semicontinuous market model. Then there exists a mapping $V:\Omega\times\R^{dT}\rightarrow\R^n\cup\{-\infty\}$, measurable with respect to $\FF\otimes\BB(\R^{dT})$, such that for each $\vartheta\in\AA$ we have
  $$
    \widehat V(\vartheta)(\omega)\sim_{K(\omega)} V(\omega,\vartheta(\omega))\qquad a.s.
  $$
\end{theorem}

To obtain a desired representation theorem, one would need to show that the construction of $V$ is independent of the choice of basis $\{Z_k\,|\,k=1,\ldots,n\}$ used to construct it. This, however, requires some additional properties of the market model $\widehat V$, like concavity, or of the representations $V_Z$ like that it is a Carath\'eodory integrand.

We now turn to the definition of the no-arbitrage condition. Above we have defined the no-arbitrage condition in terms of the recession cone. This one is not available here, hence we define the no-arbitrage condition in terms of scalarizations.
\begin{definition}
  We say that the market model $V$ satisfies the no-arbitrage condition if
$$
  \big\{\vartheta\in\AA\,\big|\,V_Z^\infty(\vartheta)\geq0\,\,a.s.\,\,\forall Z\in L^0(\FF;\textnormal{ri}\,K^\circ)\big\} = \{0\}.
$$
\end{definition}

\begin{remark}\label{ex:recession}
An obvious question at this point would be how to check that a market model satisfies the no-arbitrage condition. Here we give two simple situations in which the property is direct.
\begin{enumerate}
  \item
  Let $V$ be an upper-semicontinuous market model that is positively homogeneous. This means that $V(\lambda\vartheta)=\lambda V(\vartheta)$ for each $\lambda\geq0$. Then the definition of no-arbitrage reduces to the classical efficient no-arbitrage condition (cf.~\cite{kabanov2002no}): A market model $V$ satisfies the no-arbitrage condition if
$$
  V(\vartheta)\succeq_K 0\quad\Longrightarrow\quad \vartheta=0.
$$
Indeed, for every selection $Z\in L^0(\FF;\textnormal{ri}\,K^\circ)$, the function $V_Z$ is positively homogeneous and upper-semicontinuous by assumptions. Therefore also $V_Z^\infty=V_Z$.
\item
  Let $V$ be an upper-semicontinuous model for which there exists a random vector $\zeta\in L^0(\FF;\R^n)$ and a matrix $L\in L^0(\FF;\R^{n\times dT})$, such that
   \begin{enumerate}
     \item $\zeta+ L\vartheta\succeq_K V(\vartheta)$ a.s. for all $\vartheta\in\AA$; and
     \item the market model $\vartheta\mapsto L\vartheta$ satisfies the no-arbitrage condition.
   \end{enumerate}
  Then also $V$ satisfies the no-arbitrage condition. Indeed, for every $Z\in L^0(\FF;\textnormal{ri}\,K^\circ)$ we have
  $$
    V_Z^\infty(\vartheta)=\langle Z,V\rangle^\infty(\vartheta)\leq \langle Z,\zeta+ L\,\cdot\,\rangle^\infty(\vartheta) = \langle Z,L\vartheta\rangle,
  $$
  where by $\langle Z,\zeta+ L\,\cdot\,\rangle$ we denoted the function $x\mapsto\langle Z,\zeta+ Lx\rangle$.
  A variant of this condition apeared in~\cite{bouchard2011no} as \emph{no marginal arbitrage (of the second kind) for high production regimes}.
\end{enumerate}
\end{remark}

We now turn to explain how the results from the one-dimensional case transfer to this setup. For a random vector $f\in L^0(\FF;\R^n)$ define a set of strategies that dominate it
$$
  \AA_f = \big\{\vartheta\in\AA\,\big|\,V(\vartheta)\succeq_K f\,\,a.s.\big\}.
$$

\begin{lemma}
  Let $V$ be a market model satisfying the no-arbitrage condition. Then the set $\AA_f$ is bounded in probability.
\end{lemma}
\begin{proof}
The proof is almost identical to the one dimensional case, Lemma~\ref{lem:closure}. Let $(\vartheta_n)$ be an unbounded sequence in $\AA_f$. The crux of the proof of Lemma~\ref{lem:closure} is in extracting a random subsequence $\psi_n$, not necessarily adapted, of the sequence $\vartheta_n$ such that $\frac{\psi_n}{1+|\psi_n|}\rightarrow\psi$ almost surely and also such that $|\psi|\in\{0,1\}$. Denote $A=\{\omega\,|\,\lim_n|\psi_n|=\infty\}$ and notice that $\psi\indic_A = \psi$ by the assumption on the values of $|\psi|$. Then for every selection $Z\in L^0(\FF;\textrm{ri}\,K^\circ)$ we have 
\begin{align*}
  V_Z^\infty(\psi)
  &\geq \limsup_{n\rightarrow\infty}\left\langle Z,\frac{1}{1+|\psi_n|}V(\psi_n)\indic_A\right\rangle
  \\&\geq \limsup_{n\rightarrow\infty}\left\langle Z,\frac{1}{1+|\psi_n|}f\indic_A\right\rangle=0,
\end{align*}
i.e. $\psi$ is an arbitrage strategy.
\end{proof}

\begin{theorem}\label{thm:superhedging bound}
  Let $V$ be a market model satisfying the no-arbitrage condition. For every random vector $f\in L^0(\FF;\R^n)$ there exists a random variable $K_f$ such that $\vartheta\in\AA_f$ implies $|\vartheta|\leq K_f$ almost surely.
\end{theorem}
\begin{proof}
  We have proved the equivalent claim for the case of $n=1$; this general case does not require any new ideas. Since the no-arbitrage condition is defined in terms of selections
$$
  V\,\textrm{ satisfies NA }\quad\Longleftrightarrow\quad \bigcap_{Z\in L^0(\textrm{ri}\,K^\circ)}\big\{\vartheta\in\AA\,\big|\,V_Z^\infty(\vartheta)\geq0\big\} = \{0\},
$$
we will need to formulate the proof differently.

Define the filtration $\F_0=(\FF_0,\ldots,\FF_{T-1})$ and for each selection $Z\in L^0(\textrm{ri}\,K^\circ)$ the set of market models $V_Z^0(\vartheta) = V_Z(\vartheta)$. So, we also modify 
$$
  \AA_f^0 = \big\{\vartheta\,\textrm{ adapted to }\F_0\,\big|\,V_Z^0(\vartheta)\geq \langle Z,f\rangle\,\,a.s.\,\,\,\forall Z\in L^0(\FF;\textrm{ri}\,K^\circ)\big\}.
$$
Define the random variable 
$$
  m_0 = \esssup\left\{|\vartheta_0|\,|\,\vartheta\in\AA^0_f\right\}.
$$
As we are taking essential supremum over a set that is upward directed, by boundedness in probability of the set $\AA^0_f$, we necessarily have $m_0<\infty$  a.s.

Now to the 'induction step'. Define the filtration $\F_1=(\FF_1,\FF_1,\FF_2,\ldots,\FF_{T-1})$, i.e. first two sigma algebras are equal to $\FF_1$, the rest stays the same. Define the market models $V_Z^2(\vartheta) = V_Z^2(\vartheta) - \indic_{|\vartheta_0|\leq m_0}$, where the indicator is the convex-analytic one: takes the value 0 if the argument is true and $\infty$ otherwise. So, we also modify 
$$
  \AA_f^1 = \big\{\vartheta\,\textrm{ adapted to }\F_1\,\big|\,V_Z^1(\vartheta)\geq \langle Z,f\rangle\,\,a.s.\,\,\,\forall Z\in L^0(\FF;\textrm{ri}\,K^\circ)\big\}.
$$
Before proceeding, we check the no-arbitrage condition. Assume that $\vartheta$ is an arbitrage oportunity for the new market model, i.e. $(V_Z^1)^\infty(\vartheta)\geq0$ a.s. for all $Z\in L^0(\textrm{ri}\,K^\circ)$. But then, by the definition of $V_Z^1$, we necessarily have $\vartheta_0=0$. This implies that $\vartheta$ is $\F_0$ adapted, hence also $(V_Z^0)^\infty(\vartheta)=(V_Z^1)^\infty(\vartheta)\geq0$ a.s., i.e. $\vartheta$ is an arbitrage strategy for the market $V$, which is a contradiction to the assumptions. Now, the previous lemma implies that $\AA_f^1$ is bounded in probability.
 Define the random variable 
$$
  m_1 = \esssup\left\{|\vartheta_1|\,|\,\vartheta\in\AA^0_f\right\},
$$
which is, again, an essential supremum over a set that is upward directed, hence $m_1<\infty$  a.s.

Repeat this step for all $t$ up to $t=T-1$ and then set $K_f = m_0 +\cdots+m_{T-1}$.
\end{proof}

Note that in the proof we have called a collection $\{V^i_Z\,|\,Z\in L^0(\FF;\textrm{ri}\,K^\circ)\}$ a market model and used claims above for the vector valued market model. This is legitimate, since we are defining all concepts in terms of scalarizations. Therefore, we may `forget' that behind the collection $\{\langle Z,V\rangle\,|\,Z\in L^0(\FF;\textrm{ri}\,K^\circ)\}$ there is a market model $V:\Omega\times\R^{dT}\rightarrow\R^n$.

We now state the classic theorem that no-arbitrage implies the closedness of the set of superhedgeable claims
$$
  \CC = \big\{h\in L^0(\FF;\R^n)\,\big|\,\exists\vartheta\in\AA:\,\,\widehat V(\vartheta)\succeq_K h\,\,\, a.s.\big\}.
$$
We also define the set $\CC_f = \{h\in\CC\,|\,h\succeq_Kf\,\,a.s.\}$.

\begin{lemma}
  If a market model satisfies the no-arbitrage condition, then $\CC_f$ is bounded in probability.
\end{lemma}
\begin{proof}
  The proof proceeds by the same line of argument as in the $n=1$ case.
\end{proof}

\begin{theorem}
  Let $\widehat V$ be an upper-semicontinuous market model satisfying the no-arbitrage condition. Then the set $\CC$ is closed in probability.
\end{theorem}
\begin{proof}
  Let $h_n$ be a sequence in $\CC$ that converges in probability to a random vector $h$. Passing to a subsequence, we may assume that the convergence takes place almost surely. 

  \textsc{Step 1:} there exists a $g\in L^0(\FF;\R^n)$ such that $h_n\succeq_K g$ a.s. for all $n$. Since $h_n\rightarrow h$ a.s. then the random variable $\rho = \sup_n\|h_n\|$ is finite. Now, choose a measurable selection $x\in L^0(\FF;\mathrm{int}\,K)$ and an $\FF$ measurable random variable $r>0$, such that $B_r(x)\subset \mathrm{int}\,K$; this exists by Lemma~\ref{lemma:zasto ne radi}. We can choose $g = -\frac{\rho}{r}x$. 

\textsc{Step 2:} Let $\vartheta_n\subset\AA$ be a sequence of strategies with $V(\vartheta_n)\succeq_K h_n$. By Step 1 and Theorem~\ref{thm:superhedging bound}, we know that the sequence $(\vartheta_n)$ is bounded in the sense $|\vartheta_n|\leq K_g$. As in the proof of the case with $n=1$, we may choose a random subsequence $\vartheta_{\tau(n)}$ such that $\lim_n\vartheta_{\tau(n)}=\vartheta\in\AA$ and $\tau(n):\Omega\rightarrow\N$ is an increasing sequence of random variables. The result now follows from the observation that $V$ is upper-semicontinuous and $V(\vartheta_{\tau(n)})\succeq_K h_{\tau(n)}$. So, $V(\vartheta)\succeq_K h$.
\end{proof}

We give here a simple model of a financial market that is in the style of a market model presented by~\cite{bouchard2006no}. It is the basic model of a market with transaction costs. The modelling approach comes from the first paper on currency markets~\cite{kabanov1999hedging}.

\begin{example}\label{example of bouchard}
Consider the financial market model with $d$ assets. The portfolio of the trader is modeled as a sequence of random vectors $(V_t)_{i=0}^T$, each component of which specifies the number of shares of the asset the trader holds in the portfolio. 

Trading is included into the market through the sequence of orders. Denote by $\textbf{M}^d$ the set of matrices with zero diagonal. Trading in the market will be encoded with elements of $\textbf{M}^d$, where the number in the slot $(i,j)$ denotes the number of shares of asset $i$ that is transfered to asset $j$. Negative entry implies that shares of asset $i$ are bought and paid from asset $j$. Denote the space of $\F$-adapted, $\textbf{M}^d$-valued sequences with $\AA$. 

Part of the volume traded is absorbed by the market as transaction costs. We will denote by $F_t:\Omega\times\textbf{M}^d\rightarrow\R^d$ the change of portfolio of the trader as a consequence of executing the order $\vartheta_{t-1}$, i.e. for $\vartheta\in\AA$ we denote $V_t-V_{t-1} = F_t(\vartheta_{t-1})$. Our market model is then $V(\vartheta)=\sum_{i=1}^T F_i(\vartheta_{i-1})$.
We will assume that
\begin{enumerate}
  \item[(i)] for each fixed $\omega$, the mapping $x\mapsto F_t(\omega,x)$ is continuous; and
  \item[(ii)] for each fixed $x$, the mapping $\omega\mapsto F_t(\omega,x)$ is $\FF$ measurable.
\end{enumerate}
In other words, we assume that $F_t$ are $\FF$ Carath\'eodory integrands; this in turn implies that $V$ is an upper-semicontinuous market model. 

In order to talk about the no-arbitrage condition, let us specify order on $\R^d$ by defining the order cone as $K=\R^d_+$. The no-arbitrage condition is expressed in terms of the recession market model and one can easily convince oneself that this is given as $V^\infty(\vartheta)=\sum_{i=1}^T F_i^\infty(\vartheta_{i-1})$. The no-arbitrage condition is then expressed as: $V^\infty(\vartheta)\succeq_{\R^d_+}0$ implies that $\vartheta=0$. Compare this to the weak no-arbitrage condition of~\cite{bouchard2006no} for the case of concave positively homogeneous market model.
\end{example}

\appendix
\section{On measurable correspondences}
\label{ap:prvi}

Here we list some general information about measurable correspondences. The main reference for this is Chapter~14 in~\cite{rockafellar1998variational}.

A set valued mapping $A:\Omega\rightrightarrows\R^n$ is called an $\FF$ measurable correspondence if for every open set $V\subset\R^n$ the set $\{\omega\in\Omega\,|\,A(\omega)\cap V\not=\varnothing\}$ is in $\FF$. When $A$ is single valued, i.e. $A(\omega)$ is a singleton for almost all $\omega$, this definition coincides with the definition of a random vector.

If $A$ is an $\FF$ measurable correspondence, then also the closure $\omega\mapsto\overline{A(\omega)}$ is an $\FF$ measurable correspondence; see Proposition~14.2 in~\cite{rockafellar1998variational}.

A measurable selection $\phi$ of a measurable correspondence $A$ is a random vector $\phi\in L^0(\FF;\R^n)$ such that $\phi(\omega)\in A(\omega)$ for all $\omega$. We denote the set of all $\FF$ measurable selections of a correspondence $A$ by $L^0(\FF;A)$.
 
A Castaing representation $\c$ of a closed-valued measurable correspondence $A$ is a countable set $\c=\{\phi_k\,|\,k\in\N\}$ of measurable selections of the correspondence $A$, such that
$$
A(\omega) = \overline{\{\phi_k(\omega)\,|\,\phi_k\in\c\}}\quad\textrm{ for all }\omega.
$$
A Castaing representation of an $\FF$ measurable, closed valued correspondence always exists; see Theorem~14.5 in~\cite{rockafellar1998variational}.

\begin{lemma}
A convex valued correspondence $A:\Omega\rightrightarrows\R^n$ admits a measurable selection $\rho$ of the relative interior of $A$. 
\end{lemma}
\begin{proof}
By above, the measurable correspondence $\overline{A}$ admits a measurable selection $\psi$. Choose a Castaing representation $\c=\{\phi_n\,|\,n\in\N\}$ for the measurable correspondence $\overline{A}\cap B_{\|\psi\|+1}(0)$. One may now simply check that the random vector $\rho=\sum_{k=1}^\infty 2^{-k}\phi_k$ is the sought for measurable selection.
\end{proof}

\begin{corollary}
  A convex valued correspondence $A:\Omega\rightrightarrows\R^n$ admits a Castaing representation $\c$ with all elements $\phi\in\c$ selections for the relative interior of $A$.
\end{corollary}
\begin{proof}
  Let $\c$ be the Castaing representation of the correspondence $\omega\mapsto\overline{A(\omega)}$ and let $\rho$ be a measurable selection of $\textrm{ri}\,A$. Then the set 
$$
  \widehat\c = \left\{\frac1n\rho + \left(1-\frac1n\right)\phi\,\Big|\,\phi\in\c,\,\,n\in\N\right\}
$$
has the desired properties. Note, in particular, that $\ri A\subseteq A\subseteq\overline A$.
\end{proof}

\begin{corollary}\label{corollary:relative interior}
  Let $A$ be a convex valued measurable correspondence. Then the correspondence $\textrm{ri}\,A$ is measurable.
\end{corollary}
\begin{proof}
  Let $\c$ be a Castaing representation with $\c\subset L^0(\FF;\textrm{ri}\,A)$. Then by Exercise 14.12(a) in~\cite{rockafellar1998variational} also its convex hull
$$
  \omega\mapsto\textrm{conv}\{\phi(\omega)\,|\,\phi\in\c\}=\textrm{ri}\,A(\omega)
$$
is measurable.
\end{proof}

\begin{lemma}\label{lemma:zasto ne radi}
  Let $A:\Omega\rightrightarrows\R^n$ be a convex valued $\FF$ measurable correspondence. Let $\rho$ be a measurable selection of $\textnormal{int}\,A$, which we assume to be non-empty. Then there exists a random variable $r>0$ such that $B_{r(\omega)}(\rho(\omega))\subset A(\omega)$ a.s.
\end{lemma}
\begin{proof}
  To prove the existence of the random variable $r$, note that for any $r$ such that $B_r(x)\subset A$, by the triangle inequality, we have for every $q\in\Q^n$
$$
  r\leq |q-x|-d(q,A) + \infty\indic_{A}(q),
$$
where $d(q,K)$ is the distance of the point $q$ to the correspondence $A$. The indicator is the classical one; it has value 1 if the argument is true and 0 otherwise. This term in the estimate above says that for each $\omega$ we consider only $q$ that are not in $A$. The expression on the right hand side is $\FF$ measurable by~\cite{rockafellar1998variational}, Theorem 14.3 (j) and Example 14.7. So, 
$$
  r = \frac12 \inf_{q\in\Q^n}\left[ |q-x|-d(q,A)+\infty\indic_A(q)\right]
$$
satisfies the desired properties.
\end{proof}

An analogous statement can be shown also for convex correspondences that do not have an interior.
\begin{lemma}\label{lemma relative interior distance}
  Let $A$ be a convex valued correspondence and let $\rho$ be the measurable selection of $\ri A$. Then there exists a random variable $r>0$, such that $$B_{r(\omega)}(\rho(\omega))\cap\textnormal{aff}\,A(\omega)\subset A(\omega)$$ for almost all $\omega$, where by $\textnormal{aff}\,A(\omega)$ we denoted the affine hull of $A$.
\end{lemma}
\begin{proof}
Note first, that the affine hull $\textnormal{aff}\,A$ is a measurable correspondence by Exercise 14.12(c) in~\cite{rockafellar1998variational}. Then also correspondence $\widehat A(\omega) = \textnormal{aff}\,A(\omega) - \rho(\omega)$ is measurable by Proposition~14.11(c) in~\cite{rockafellar1998variational}. Also, the values of the correspondence $\widehat A$ are linear subspaces of $\R^n$. Thus also the correspondence $\widehat A^\perp$, whose value for each $\omega$ is the orthogonal complement of $\widehat A(\omega)$, is measurable by Exercise~14.12(f) in~\cite{rockafellar1998variational}. The statement of the lemma now follows by applying the previous lemma to the measurable correspondence $A+\widehat A^\perp$.
\end{proof}


\bibliographystyle{plain}
\bibliography{literature}

\begin{thebibliography}{10}

\bibitem{bertsekas1974necessary}
Dimitri~P Bertsekas.
\newblock Necessary and sufficient conditions for existence of an optimal
  portfolio.
\newblock {\em Journal of Economic Theory}, 8(2):235--247, 1974.

\bibitem{bouchard2006no}
Bruno Bouchard.
\newblock No-arbitrage in discrete-time markets with proportional transaction
  costs and general information structure.
\newblock {\em Finance and Stochastics}, 10(2):276--297, 2006.

\bibitem{bouchard2011no}
Bruno Bouchard and Adrien Nguyen~Huu.
\newblock No marginal arbitrage of the second kind for high production regimes
  in discrete time production--investment models with proportional transaction
  costs.
\newblock {\em Mathematical Finance}, 2011.

\bibitem{cheridito2012conditional}
Patrick Cheridito, Michael Kupper, and Nicolas Vogelpoth.
\newblock Conditional analysis on {$\R^d$}.
\newblock {\em arXiv preprint arXiv:1211.0747}, 2012.

\bibitem{dolinsky2014risk}
Yan Dolinsky and Yuri Kifer.
\newblock Risk minimization in markets imposing minimal transaction costs.
\newblock {\em arXiv preprint arXiv:1408.3774}, 2014.

\bibitem{dolinsky2013duality}
Yan Dolinsky and Halil~Mete Soner.
\newblock Duality and convergence for binomial markets with friction.
\newblock {\em Finance and Stochastics}, pages 1--29, 2013.

\bibitem{drapeau2013brouwer}
Samuel Drapeau, Martin Karliczek, Michael Kupper, and Martin Streckfu{\ss}.
\newblock Brouwer fixed point theorem in $(l^0)^d$.
\newblock {\em Fixed Point Theory and Applications}, 2013(1):1--14, 2013.

\bibitem{evstigneev1976measurable}
IV~Evstigneev.
\newblock Measurable selection and dynamic programming.
\newblock {\em Mathematics of Operations Research}, 1(3):267--272, 1976.

\bibitem{follmer2011stochastic}
Hans F{\"o}llmer and Alexander Schied.
\newblock {\em Stochastic finance: an introduction in discrete time}.
\newblock Walter de Gruyter, 2011.

\bibitem{kabanov2013essential}
Yuri Kabanov and Emmanuel L{\'e}pinette.
\newblock Essential supremum with respect to a random partial order.
\newblock {\em Journal of Mathematical Economics}, 49(6):478--487, 2013.

\bibitem{kabanov2002no}
Yuri Kabanov, Mikl{\'o}s R{\'a}sonyi, and Christophe Stricker.
\newblock No-arbitrage criteria for financial markets with efficient friction.
\newblock {\em Finance and Stochastics}, 6(3):371--382, 2002.

\bibitem{kabanov1999hedging}
Yuri~M Kabanov.
\newblock Hedging and liquidation under transaction costs in currency markets.
\newblock {\em Finance and Stochastics}, 3(2):237--248, 1999.

\bibitem{lobo2007portfolio}
Miguel~Sousa Lobo, Maryam Fazel, and Stephen Boyd.
\newblock Portfolio optimization with linear and fixed transaction costs.
\newblock {\em Annals of Operations Research}, 152(1):341--365, 2007.

\bibitem{pennanen2011arbitrage}
Teemu Pennanen.
\newblock Arbitrage and deflators in illiquid markets.
\newblock {\em Finance and Stochastics}, 15(1):57--83, 2011.

\bibitem{pennanen2011convex}
Teemu Pennanen.
\newblock Convex duality in stochastic optimization and mathematical finance.
\newblock {\em Mathematics of Operations Research}, 36(2):340--362, 2011.

\bibitem{pennanen2010hedging}
Teemu Pennanen and Irina Penner.
\newblock Hedging of claims with physical delivery under convex transaction
  costs.
\newblock {\em SIAM Journal on Financial Mathematics}, 1(1):158--178, 2010.

\bibitem{penannen2015non}
Teemu Pennanen, Ari-Pekka Perkki{\"o}, and Mikl{\'o}s R{\'a}sonyi.
\newblock Non-convex dynamic programming and optimal investment.
\newblock {\em arXiv preprint arXiv:1504.01903}, 2015.

\bibitem{peskir2006optimal}
Goran Peskir and Albert Shiryaev.
\newblock {\em Optimal stopping and free-boundary problems}.
\newblock Springer, 2006.

\bibitem{roch2011resilient}
Alexandre Roch and H~Mete Soner.
\newblock Resilient price impact of trading and the cost of illiquidity.
\newblock {\em Preprint. URL: http://ssrn. com/paper}, 1923840, 2011.

\bibitem{rockafellar1998variational}
R~Tyrrell Rockafellar and Roger J-B Wets.
\newblock {\em Variational Analysis}, volume 317.
\newblock Springer Verlag, 2004.

\bibitem{schachermayer1992hilbert}
Walter Schachermayer.
\newblock A hilbert space proof of the fundamental theorem of asset pricing in
  finite discrete time.
\newblock {\em Insurance: Mathematics and Economics}, 11(4):249--257, 1992.

\bibitem{vzitkovic2002filtered}
Gordan {\v{Z}}itkovi{\'c}.
\newblock A filtered version of the bipolar theorem of brannath and
  schachermayer.
\newblock {\em Journal of Theoretical Probability}, 15(1):41--61, 2002.

\end{thebibliography}

\end{document}